\documentclass[english,twocolumn, fleqn, secthm, amsthm]{autart}
\usepackage[T1]{fontenc}
\usepackage[latin9]{inputenc}
\usepackage{amsmath}
\usepackage{amssymb}
\usepackage{graphicx}

\makeatletter
\usepackage[harvard,dcucite]{harvard}
\usepackage{balance}
\usepackage{dsfont}
\newcommand{\mathbbm}[1]{\mathds{#1}}

\usepackage{times} 

\makeatother

\usepackage{babel}
\begin{document}
\begin{frontmatter} 
\title{Feedback Control of Switched Stochastic Systems Using \\   Randomly Available Active Mode Information \thanksref{footnoteinfo}} \thanks[footnoteinfo]{This research was supported in part by JSPS Grant-in-Aid for Scientific Research (A) 26249062 and (C) 25420431, the Aihara Innovative Mathematical Modelling Project (JSPS) under FIRST program initiated by CSTP, and Japan Science and Technology Agency under CREST program. The material in this paper was partially presented at the 52nd IEEE Conference on Decision and Control, 2013, Firenze, Italy.}
\author[titech]{Ahmet Cetinkaya} \ead{ahmet@dsl.mei.titech.ac.jp},
\author[titech]{Tomohisa Hayakawa \thanksref{telfax}} \ead{hayakawa@mei.titech.ac.jp} \thanks[telfax]{Tel. : +81 3 5734 2762; Fax: +81 3 5734 2762 } 
\address[titech]{Department of Mechanical and Environmental Informatics, Tokyo Institute of Technology, Tokyo 152-8552, Japan}

\begin{abstract} 

Almost sure asymptotic stabilization of a discrete-time switched stochastic
system is investigated. Information on the active operation mode of
the switched system is assumed to be available for control purposes
only at random time instants. We propose a stabilizing feedback control
framework that utilizes the information obtained through mode observations.
We first consider the case where stochastic properties of mode observation
instants are fully known. We obtain sufficient asymptotic stabilization
conditions for the closed-loop switched stochastic system under our
proposed control law. We then explore the case where exact knowledge
of the stochastic properties of mode observation instants is not available.
We present a set of alternative stabilization conditions for this
case. The results for both cases are predicated on the analysis of
a sequence-valued process that encapsulates the stochastic nature
of the evolution of active operation mode between mode observation
instants. Finally, we demonstrate the efficacy of our results with
numerical examples. 

\end{abstract} 
\end{frontmatter}

\section{Introduction}

The framework developed for switched stochastic systems provides accurate
characterization of numerous complex real life processes from physics
and engineering fields that are subject to randomly occurring incidents
such as sudden environmental variations or sharp dynamical changes
\cite{cassandras2006,yinzhu2010}. Stabilization problem for switched
stochastic systems has been investigated in many studies (e.g., \citeasnoun{ghaoui1996}, \citeasnoun{farias2000}, \citeasnoun{fang2002}, \citeasnoun{costa2004discrete}, \citeasnoun{sathananantan2008}, \citeasnoun{geromel2009}
and the references therein). 

Control frameworks developed for switched stochastic systems often
require the availability of information on the active operation mode
at all times. Note that for numerous applications the active mode
describes the operating conditions of a physical process and is driven
by external incidents of stochastic nature. The active mode, hence,
may not be directly measurable and it may not be available for control
purposes at all time instants during the course of operation. When
the controller does not have access to any mode information, for achieving
stabilization one can resort to adaptive control frameworks \cite{toussi1991,caines1992,bercu2009}
or mode-independent control laws \cite{vargas2006,boukasautomatica2006}.
On the other hand, if mode information can be observed at certain
time instants (even if rarely), this information can be utilized in
the control framework. In our earlier work \cite{cetinkayaacc2012,cetinkaya2013a},
we investigated stabilization of switched stochastic systems for the
case where only \emph{sampled} mode information is available for control
purposes. Under the assumption that the active mode is \emph{periodically}
observed, we proposed a stabilizing feedback control framework that
utilizes the available mode information. 

In practical applications, it would be ideal if the mode information
of a switched system is available for control purposes at all time
instants or at least periodically. However, there are cases where
mode information is obtained at \emph{random} time instants. This
situation occurs for example when the mode is sampled at all time
instants; however, some of the mode samples are randomly lost during
communication between mode sampling mechanism and the controller.
On the other hand, in some applications, the mode has to be detected,
but the detected mode information may not always be accurate. In this
case each mode detection has a confidence level. Mode information
with low confidence is discarded. As a result, depending on the confidence
level of detection, the controller may or may not receive the mode
information at a particular mode detection instant. In addition, we
may also take advantage of random sampling for certain cases and observe
the mode intentionally at random instants, as for such cases control
under random sampling provides better results compared to periodic
sampling. Note that random sampling has also been used for problems
such as signal reconstruction and has been shown to have advantages
over regular periodic sampling (see \citeasnoun{boyle2007detecting}, \citeasnoun{carlen2009signal}).

In this paper our goal is to explore the feedback stabilization problem
for the case where the active operation mode, which is modeled as
a finite-state Markov chain, is observed at \emph{random} time instants.
We provide an extended discussion based on our preliminary report
\cite{cetinkaya2013cdc}. Specifically, we assume that the length
of intervals between consecutive mode observation instants are identically
distributed independent random variables. We employ a renewal process
to characterize the occurrences of random mode observations. This
characterization allows us to also explore periodic mode observations
\cite{cetinkayaacc2012,cetinkaya2013a} as a special case. 

We propose a linear feedback control law with a piecewise-constant
gain matrix that is switched depending on the value of a randomly
sampled version of the mode signal. In order to investigate the evolution
of the active mode together with its randomly sampled version, we
construct a stochastic process that represents sequences of values
the mode takes between random mode observation instants. This sequence-valued
stochastic process turns out to be a countable-state Markov chain
defined over a set that is composed of all possible mode sequences
of finite length. We first analyze the probabilistic dynamics of this
sequence-valued Markov chain. Then based on our analysis, we obtain
sufficient stabilization conditions for the closed-loop switched stochastic
system under our proposed control framework. These stabilization conditions
let us assess whether the closed-loop system is stable for a given
probability distribution for the length of intervals between consecutive
mode observation instants. As this probability distribution is not
assumed to have a certain structure, the result presented in this
paper can also be considered as a generalization of the result provided
in \citeasnoun{cetinkaya2011}, where stabilization problem is discussed
in continuous time and the random intervals between mode sampling
instants are specifically assumed to be exponentially distributed.
In this paper we also explore the case where perfect information regarding
the probability distribution for the length of intervals between consecutive
mode observation instants is not available. For this problem setting,
we present alternative sufficient stabilization conditions which can
be used for verifying stability even if the distribution is not exactly
known. 

The paper is organized as follows. We provide the notation and a review
of key results concerning renewal processes in Section~\ref{sec:Mathematical-Preliminaries}.
In Section~\ref{sec:SwitchedStochasticSection3}, we propose our
feedback control framework for stabilizing discrete-time switched
stochastic systems under randomly available mode information. Then
in Section~\ref{sec:Sufficient-Conditions-for}, we present sufficient
conditions under which our proposed control law guarantees almost
sure asymptotic stabilization. In Section~\ref{sec:Illustrative-Numerical-Example},
we demonstrate the efficacy of our results with two illustrative numerical
examples. Finally, in Section~\ref{sec:Conclusion} we conclude our
paper.

\section{Mathematical Preliminaries\label{sec:Mathematical-Preliminaries}}

In this section, we provide notation and several definitions concerning
discrete-time stochastic processes. Specifically, we denote positive
and nonnegative integers by $\mathbb{N}$ and $\mathbb{N}_{0}$, respectively.
Moreover, $\mathbb{R}$ denotes the set of real numbers, $\mathbb{R}^{n}$
denotes the set of $n\times1$ real column vectors, and $\mathbb{R}^{n\times m}$
denotes the set of $n\times m$ real matrices. We write $(\cdot)^{\mathrm{T}}$
for transpose, $\|\cdot\|$ for the Euclidean vector norm. We use
$\lambda_{\min}(H)$ (resp., $\lambda_{\max}(H)$) for the minimum
(resp., maximum) eigenvalue of the Hermitian matrix $H$. A function
$V:\mathbb{R}^{n}\rightarrow\mathbb{R}$ is called positive definite
if $V(x)>0,\, x\neq0$, and $V(0)=0$. We represent a finite-length
sequence of ordered elements $q_{1},q_{2},\ldots,q_{n}$ by $q=(q_{1},q_{2},\ldots,q_{n})$.
The length (number of elements) of the sequence $q$ is denoted by
$|q|$. The notations $\mathrm{\mathbb{P}}[\cdot]$ and $\mathbb{E}[\cdot]$
respectively denote the probability and expectation on a probability
space $(\Omega,\mathcal{F},\mathbb{P})$ with filtration $\{\mathcal{F}_{k}\}_{k\in\mathbb{N}_{0}}$.
Furthermore, we write $\mathbbm{1}_{[G]}:\Omega\to\{0,1\}$ for the
indicator of the set $G\in\mathcal{F}$, that is, $\mathbbm{1}_{[G]}(\omega)=1$,
$\omega\in G$, and $\mathbbm{1}_{[G]}(\omega)=0$, $\omega\notin G$.

\subsection{Discrete-Time Renewal Processes\label{sub:Discrete-Time-Renewal-Processes}}

A discrete-time renewal process $\{N(k)\in\mathbb{N}_{0}\}_{k\in\mathbb{N}_{0}}$
with initial value $N(0)=0$ is an $\mathcal{F}_{k}$-adapted stochastic
counting process defined by $N(k)\triangleq\sum_{i\in\mathbb{N}}\mathbbm{1}_{[t_{i}\leq k]},$
where $t_{i}\in\mathbb{N}_{0}$, $i\in\mathbb{N}_{0}$, are random
time instants such that $t_{0}=0$ and $\tau_{i}\triangleq t_{i}-t_{i-1}\in\mathbb{N}$,
$i\in\mathbb{N}$, are identically distributed independent random
variables with finite expectation (i.e., $\mathbb{E}[\tau_{i}]<\infty$,
$i\in\mathbb{N}$). Note that $\tau_{i}$, $i\in\mathbb{N}$, denote
the lengths of intervals between time instants $t_{i}$, $i\in\mathbb{N}_{0}$.
Furthermore, we use $\mu:\mathbb{N}\to[0,1]$ to denote the common
distribution of the random variables $\tau_{i}$, $i\in\mathbb{N}$,
such that 
\begin{align}
\mathbb{P}[\tau_{i}=\tau] & =\mu_{\tau},\quad\tau\in\mathbb{N},\quad i\in\mathbb{N},
\end{align}
where $\mu_{\tau}\in[0,1]$. Note that $\sum_{\tau\in\mathbb{N}}\mu_{\tau}=1$.
Now, let $\hat{\tau}\triangleq\sum_{\tau\in\mathbb{N}}\tau\mu_{\tau}=\mathbb{E}[\tau_{1}]$($=\mathbb{E}[\tau_{i}]$,
$i\in\mathbb{N}$). It follows as a consequence of strong law of large
numbers for renewal processes (see \citeasnoun{serfozo2009}) that
$\lim_{k\to\infty}\frac{N(k)}{k}=\frac{1}{\hat{\tau}}$. 

Note that in Section~\ref{sec:SwitchedStochasticSection3}, we employ
a renewal process to characterize the occurrences of random mode observations.

\subsection{Almost Sure Asymptotic Stability}

The zero solution $x(k)\equiv0$ of a stochastic system is \emph{almost
surely stable} if, for all $\epsilon>0$ and $\rho>0$, there exists
$\delta=\delta(\epsilon,\rho)>0$ such that if $\|x(0)\|<\delta$,
then 
\begin{align}
\mathbb{P}[\sup_{k\in\mathbb{N}_{0}}\|x(k)\|>\epsilon] & <\rho.
\end{align}
 Furthermore, the zero solution $x(k)\equiv0$ of a stochastic system
is \emph{asymptotically stable almost surely} if it is almost surely
stable and 
\begin{align}
\mathbb{P}[\lim_{k\to\infty}\|x(k)\|=0] & =1.\label{eq:definition-convergence}
\end{align}
In Sections~\ref{sec:SwitchedStochasticSection3} and \ref{sec:Sufficient-Conditions-for},
we investigate almost sure asymptotic stabilization of a switched
stochastic system.

\section{Stabilizing Switched Stochastic Systems with Randomly Available Mode
Information \label{sec:SwitchedStochasticSection3}}

In this section, we propose a feedback control framework for stabilizing
a switched stochastic system by using only the randomly available
mode information. Specifically, we consider the discrete-time switched
linear stochastic system with $M\in\mathbb{N}$ number of modes given
by 
\begin{equation}
x(k+1)=A_{r(k)}x(k)+B_{r(k)}u(k),\quad k\in\mathbb{N}_{0},\label{eq:control-system}
\end{equation}
 with the initial conditions $x(0)=x_{0}$, $r(0)=r_{0}\in\mathcal{M}\triangleq\{1,2,\ldots,M\}$,
where $x(k)\in\mathbb{R}^{n}$ and $u(k)\in\mathbb{R}^{m}$ respectively
denote the state vector and the control input; furthermore, $A_{i}\in\mathbb{R}^{n\times n},\, B_{i}\in\mathbb{R}^{n\times m},\, i\in\mathcal{M}$,
are the subsystem matrices. The mode signal $\{r(k)\in\mathcal{M}\}_{k\in\mathbb{N}_{0}}$
is assumed to be an $\mathcal{F}_{k}$-adapted, $M$-state discrete-time
Markov chain with the initial distribution denoted by $\nu:\mathcal{M}\to[0,1]$
such that $\nu_{r_{0}}=1$ and $\nu_{i}=0$, $i\neq r_{0}$. 

We use the matrix $P\in\mathbb{R}^{M\times M}$ to characterize probability
of transitions between the modes of the switched system. Specifically,
$p_{i,j}\in[0,1]$, which is the $(i,j)$th entry of the matrix $P$,
denotes the probability of a transition from mode $i$ to mode $j$.
Note that $\sum_{j\in\mathcal{M}}p_{i,j}=1$, $i\in\mathcal{M}$.
Furthermore, we use $p_{i,j}^{(l)}$ to denote $(i,j)$th entry of
the matrix $P^{l}$. Note that $p_{i,j}^{(l)}\in[0,1]$ is in fact
the $l$-step transition probability from mode $i$ to mode $j$,
that is, 
\begin{align}
p_{i,j}^{(l)} & \triangleq\mathbb{P}[r(k+l)=j|r(k)=i],\,\, l\in\mathbb{N}_{0},\,\, i,j\in\mathcal{M},\label{eq:pton}
\end{align}
with $p_{i,i}^{(0)}=1$, $i\in\mathcal{M}$, $p_{i,j}^{(0)}=0$, $i\neq j$.
Furthermore, $p_{i,j}^{(1)}=p_{i,j}$, $i,j\in\mathcal{M}$. The mode
signal can be represented using a transition diagram, which shows
possible transitions between the operation modes of the switched system.
Mode transition diagram for a switched system with two modes is shown
in Figure~\ref{Flo:modetransdiag}. 

In this paper, we assume that the mode signal is an aperiodic, irreducible
Markov chain and has the invariant distribution $\pi:\mathcal{M}\to[0,1]$. 

\begin{figure}[t]
\begin{center}\includegraphics[width=0.45\columnwidth]{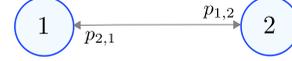}\end{center}
\vskip -5pt\protect\caption{Mode transition diagram for $\{r(k)\in\mathcal{M}\triangleq\{1,2\}\}_{k\in\mathbb{N}_{0}}$}
\label{Flo:modetransdiag}
\end{figure}

\subsection{Feedback Control Under Randomly Observed Mode Information}

In this paper, active mode of the switched stochastic system (\ref{eq:control-system})
is assumed to be observed only at random time instants, which we denote
by $t_{i}\in\mathbb{N}_{0}$, $i\in\mathbb{N}_{0}$. We assume that
$t_{0}=0$ and $\tau_{i}\triangleq t_{i}-t_{i-1}\in\mathbb{N}$, $i\in\mathbb{N}$,
are independent random variables that are distributed according to
a common distribution $\mu:\mathbb{N}\to[0,1]$ for all $i\in\mathbb{N}$
such that $\hat{\tau}\triangleq\sum_{\tau\in\mathbb{N}}\tau\mu_{\tau}<\infty$.
In this problem setting, the initial mode information $r_{0}$ is
assumed to be available to the controller, and a renewal process $\{N(k)\in\mathbb{N}_{0}\}_{k\in\mathbb{N}_{0}}$
is employed for counting the number of mode observations that are
obtained after the initial time. We assume that the renewal process
$\{N(k)\in\mathbb{N}_{0}\}_{k\in\mathbb{N}_{0}}$ and the mode signal
$\{r(k)\in\mathcal{M}\}_{k\in\mathbb{N}_{0}}$ are mutually independent.

Following our approach in \citeasnoun{cetinkaya2011}, \citeasnoun{cetinkayaacc2012}, \citeasnoun{cetinkaya2013a},
we employ a linear feedback control law with a `piecewise-constant'
feedback gain matrix that depends only on the obtained mode information.
Specifically, we consider the control law 
\begin{align}
u(k) & =K_{\sigma(k)}x(k),\quad k\in\mathbb{N}_{0},\label{eq:control-law}
\end{align}
 where $\{\sigma(k)\in\mathcal{M}\}_{k\in\mathbb{N}_{0}}$ is the
sampled version of the mode signal defined by 
\begin{align}
\sigma(k) & \triangleq r(t_{N(k)}),\quad k\in\mathbb{N}_{0}.\label{eq:sigmadef}
\end{align}
Note that the sampled mode signal $\{\sigma(k)\in\mathcal{M}\}_{k\in\mathbb{N}_{0}}$
acts as a switching mechanism for the linear feedback gain, which
remains constant between two consecutive mode observation instants,
that is, $K_{\sigma(k)}=K_{r(t_{i})}$ for $k\in[t_{i},t_{i+1})$. 

Between two consecutive mode observation instants, the feedback gain
$K_{\sigma(\cdot)}$ stays constant, whereas the active mode $r(\cdot)$
of the dynamical system (\ref{eq:control-system}) may change its
value. Stabilization performance under the control law (\ref{eq:control-law})
hence depends not only on the length of the intervals between random
mode observation instants, but also on how the active mode switches
during the intervals. 

In Figure~\ref{Flo:rsnfigure}, we show sample paths of the active
mode signal $r(\cdot)$ and its sampled version $\sigma(\cdot)$ for
a switched stochastic system with $M=2$ modes. In this example, active
mode is observed at time instants $t_{0}=0$, $t_{1}=2$, $t_{2}=5$,
$t_{3}=6$, $t_{4}=8$, $\ldots$. Note that at mode observation instants
actual mode signal $r(\cdot)$ and its sampled version $\sigma(\cdot)$
have the same value. However, at the other time instants, sampled
mode signal may differ from the actual mode, since between mode observation
instants, system mode may switch. 
\begin{figure}[t]
\begin{center}\includegraphics[width=0.85\columnwidth]{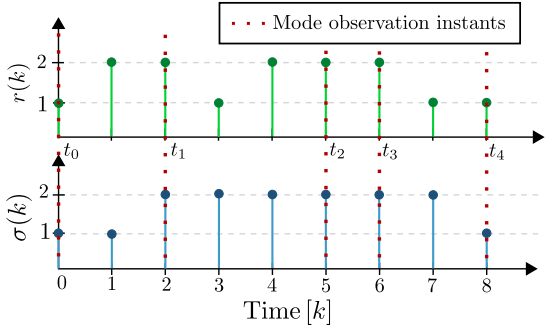}\end{center}
\vskip -5pt\protect\caption{Actual mode $r(k)$ and its sampled version $\sigma(k)$ }
\label{Flo:rsnfigure}
\end{figure}

In order to investigate the evolution of the active mode between consecutive
mode observation instants, we construct a new stochastic process $\{s(i)\}_{i\in\mathbb{N}_{0}}$
that takes values from a countable set of mode sequences of variable
length. Specifically, we define $\{s(i)\}_{i\in\mathbb{N}_{0}}$ by
\begin{align}
s(i) & \triangleq\big(r(t_{i}),r(t_{i}+1),\ldots,r(t_{i+1}-1)\big),\,\, i\in\mathbb{N}_{0},\label{eq:skdef}
\end{align}
 with $t_{i}$, $i\in\mathbb{N}_{0}$, being the random mode observation
instants. By the definition given in (\ref{eq:skdef}), $s(i)$ represents
the sequence of values that the active mode $r(\cdot$) takes between
the mode observation instants $t_{i}$ and $t_{i+1}$. Hence, $s_{n}(i)$,
which denotes the $n$th element of the sequence $s(i)$, represents
the value of the active mode $r(\cdot)$ at time $t_{i}+n-1$. Furthermore,
the value of the sampled mode signal $\sigma(\cdot)$ between time
instants $t_{i}$ and $t_{i+1}$ is represented by $s_{1}(i)=r(t_{i})$.
Note that the active mode is observed and becomes available for control
purposes only at time instants $t_{i}$, $i\in\mathbb{N}_{0}$. Thus,
the controller has access only to the observed mode data $\sigma(t_{i})=r(t_{i})$,
$i\in\mathbb{N}_{0}$, which correspond to the first elements of the
sequences $s(i)$, $i\in\mathbb{N}_{0}$. 

For the sample paths of active mode signal $r(\cdot)$ and its sampled
version $\sigma(\cdot)$ shown in Figure~\ref{Flo:rsnfigure}, mode
sequences between mode observation instants $t_{0}=0$, $t_{1}=2$,
$t_{2}=5$, $t_{3}=6$, $t_{4}=8$, are given as $s(0)=(1,2)$, $s(1)=(2,1,2)$,
$s(2)=(2)$, $s(3)=(2,1)$. The key property of the stochastic process
$\{s(i)\}_{i\in\mathbb{N}_{0}}$ is that, a given mode sequence $s(i)$
indicates full information of the active mode as well as the information
the controller has during the time interval between consecutive mode
observation instants $t_{i}$ and $t_{i+1}$. 

In what follows, we explain the probabilistic dynamics of the stochastic
process $\{s(i)\}_{i\in\mathbb{N}_{0}}$ and provide key results that
we will use in Section~\ref{sec:Sufficient-Conditions-for} for analyzing
stability of the closed-loop switched stochastic control system (\ref{eq:control-system}),
(\ref{eq:control-law}).

\subsection{Probabilistic Dynamics of Mode Sequences}

The possible values of sequence that the stochastic process $\{s(i)\}_{i\in\mathbb{N}_{0}}$
may take are characterized by the set 
\begin{align}
\mathcal{S}\triangleq\{(q_{1} & ,q_{2},\ldots,q_{\tau}):p_{q_{n},q_{n+1}}>0,\, n\in\{1,\ldots,\tau-1\};\,\,\,\,\,\,\,\,\quad\nonumber \\
 & q_{n}\in\mathcal{M},\, n\in\{1,\ldots,\tau\};\,\mu_{\tau}>0\}.
\end{align}
Note that the sequence-valued stochastic process $\{s(i)\}_{i\in\mathbb{N}_{0}}$
is a discrete-time Markov chain on the countable state space represented
by $\mathcal{S}$, which contains all possible mode sequences for
all possible lengths of intervals between consecutive mode observation
instants. For example, consider the case where the switched system
(\ref{eq:control-system}) has two modes. Furthermore, suppose that
$\mu_{\tau}>0$ for all $\tau\in\mathbb{N}$. In other words, lengths
of intervals between mode observation instants may take any positive
integer value. In this case, the state space $\mathcal{S}=\{(1),(2),(1,1),(1,2),\ldots\}$
contains all finite-length mode sequences composed of elements from
$\mathcal{M}=\{1,2\}$. See Figure~\ref{Flo:transforscountable}
for the transition diagram of countable-state Markov chain $\{s(i)\in\mathcal{S}\}_{i\in\mathbb{N}_{0}}$
of this example. 

\begin{figure}[t]
\begin{center}\includegraphics[width=0.76\columnwidth]{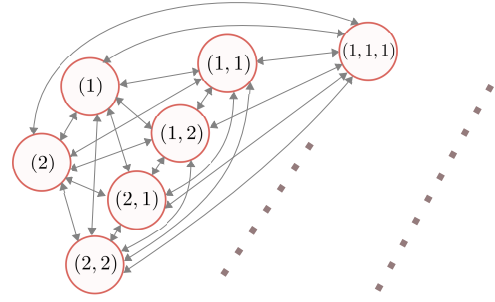}\end{center}
\vskip -5pt\protect\caption{Transition diagram of the sequence-valued discrete-time countable-state
Markov chain $\{s(i)\in\mathcal{S}\triangleq\{(1),(2),(1,1),\ldots\}\}_{i\in\mathbb{N}_{0}}$
over the set of mode sequences of variable length }
\label{Flo:transforscountable}
\end{figure}

It is important to note that if the set $\{\tau\in\mathbb{N}\,:\,\mu_{\tau}>0\}$
has finite number of elements, then set $\mathcal{S}$ will also contain
finite number of sequences. In other words, if the lengths of intervals
between mode observation instants have finite number of possible values,
then the number of possible sequences is also finite. For example,
consider the case where the operation mode of the switched system,
which takes values from the index set $\mathcal{M}=\{1,2\}$, is observed
periodically with period $2$, that is, $\mu_{2}=1$. In this case,
$\mathcal{S}=\{(1,1),(1,2),(2,1),(2,2)\}$ (see Figure~\ref{Flo:transforsfinite}). 

\begin{figure}[t]
\begin{center}\includegraphics[width=0.65\columnwidth]{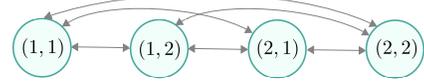}\end{center}
\vskip -5pt\protect\caption{Transition diagram of the sequence-valued discrete-time Markov chain
$\{s(i)\in\mathcal{S}\triangleq\{(1,1),(1,2),(2,1),(2,2)\}\}_{i\in\mathbb{N}_{0}}$ }
\label{Flo:transforsfinite}
\end{figure}

We now characterize the initial distribution and the state-transition
probabilities of the discrete-time Markov chain $\{s(i)\in\mathcal{S}\}_{i\in\mathbb{N}_{0}}$
as functions of the initial distribution and the state-transition
probabilities of the mode signal $\{r(k)\in\mathcal{M}\}_{k\in\mathbb{N}_{0}}$.
Specifically, the initial distribution $\lambda:\mathcal{S}\to[0,1]$
of the Markov chain $\{s(i)\in\mathcal{S}\}_{i\in\mathbb{N}_{0}}$
is given by 
\begin{align}
\lambda_{q} & =\mathbb{P}[s(0)=q]\nonumber \\
 & =\mathbb{P}[t_{1}=|q|,r(0)=q_{1},\ldots,r(|q|-1)=q_{|q|}]\nonumber \\
 & =\mathbb{P}[t_{1}=|q|\,\,\big|\,\, r(0)=q_{1},\ldots,r(|q|-1)=q_{|q|}]\nonumber \\
 & \quad\cdot\mathbb{P}[r(0)=q_{1},\ldots,r(|q|-1)=q_{|q|}],\,\, q\in\mathcal{S}.
\end{align}
Since the mode signal $\{r(k)\in\mathcal{M}\}_{k\in\mathbb{N}_{0}}$
and the mode observation counting process $\{N(k)\in\mathbb{N}_{0}\}_{k\in\mathbb{N}_{0}}$
are mutually independent, mode transitions and mode observations occur
independently. Hence, $t_{1}=\tau_{1}$ is independent of $r(n)$
for every $n\in\mathbb{N}_{0}$. As a consequence, 
\begin{align}
\lambda_{q} & =\mathbb{P}[t_{1}=|q|]\,\mathbb{P}[r(0)=q_{1},\ldots,r(|q|-1)=q_{|q|}]\nonumber \\
 & =\mathbb{P}[t_{1}=|q|]\,\mathbb{P}[r(0)=q_{1}]\nonumber \\
 & \quad\cdot\prod_{n=1}^{|q|-1}\mathbb{P}[r(n)=q_{n+1}|r(n-1)=q_{n}]\nonumber \\
 & =\begin{cases}
\mu_{|q|}\prod_{n=1}^{|q|-1}p_{q_{n},q_{n+1}}, & \quad\mathrm{if}\,\, q_{1}=r_{0},\,\, q\in\mathcal{S},\\
0, & \quad\mathrm{otherwise}.
\end{cases}\label{eq:sinitdist}
\end{align}
Note that $s_{1}(0)$, which is the first element of the first mode
sequence $s(0)$, is equal to the initial mode $r_{0}$. 

Probability of a transition from a mode sequence $q\in\mathcal{S}$
to another mode sequence $\bar{q}\in\mathcal{S}$ is given by 
\begin{align}
\rho_{q,\bar{q}} & =\mathbb{P}[s(i+1)=\bar{q}|s(i)=q],\nonumber \\
 & =\mathbb{P}\big[\tau_{i+1}=|\bar{q}|,r(t_{i+1})=\bar{q}_{1},\ldots,\nonumber \\
 & \quad\quad r(t_{i+1}+|\bar{q}|-1)=\bar{q}_{|\bar{q}|}\,\big|\,\tau_{i}=|q|,\nonumber \\
 & \quad\quad r(t_{i})=q_{1},\ldots,r(t_{i}+|q|-1)=q_{|q|}\big],
\end{align}
 for $i\in\mathbb{N}_{0}$. Note that $\tau_{i+1}$ is independent
of the random variables $r(n),\, n\in\mathbb{N}_{0}$, and $\tau_{i}$.
Furthermore, given $r(t_{i}+\tau_{i}-1)$, the random variable $r(t_{i+1})$
is conditionally independent of $r(t_{i}),\ldots,r(t_{i}+\tau_{i}-2)$,
and $\tau_{i}$. It follows that 
\begin{align}
\rho_{q,\bar{q}} & =\mathbb{P}\big[\tau_{i+1}=|\bar{q}|,r(t_{i+1})=\bar{q}_{1},\ldots,\nonumber \\
 & \quad\quad r(t_{i+1}+|\bar{q}|-1)=\bar{q}_{|\bar{q}|}\,\big|\, r(t_{i}+|q|-1)=q_{|q|}\big]\nonumber \\
 & =\mathbb{P}[r(t_{i+1})=\bar{q}_{1}\,|\, r(t_{i}+|q|-1)=q_{|q|}]\mathbb{P}[\tau_{i+1}=|\bar{q}|]\nonumber \\
 & \,\,\,\,\,\,\cdot\prod_{n=1}^{|\bar{q}|-1}\mathbb{P}[r(t_{i+1}+n)=\bar{q}_{n+1}|r(t_{i+1}+n-1)=\bar{q}_{n}]\nonumber \\
 & \,=p_{q_{|q|},\bar{q}_{1}}\,\,\,\mu_{|\bar{q}|}\,\prod_{n=1}^{|\bar{q}|-1}p_{\bar{q}_{n},\bar{q}_{n+1}},\quad i\in\mathbb{N}_{0}.\label{eq:stransprob}
\end{align}
Note that $\mu_{|\bar{q}|}$ in (\ref{eq:stransprob}) represents
the probability that length of the interval between two mode observation
instants is equal to the length of the sequence $\bar{q}$, whereas
$p_{q_{|q|},\bar{q}_{1}}\in[0,1]$ represents the transition probability
from the mode represented by the last element of sequence $q$, to
the mode represented by the first element of the sequence $\bar{q}$.
Furthermore, the expression $\prod_{n=1}^{|\bar{q}|-1}p_{\bar{q}_{n},\bar{q}_{n+1}}$
denotes the joint probability that the active mode takes the values
denoted by the elements of the sequence $\bar{q}$ until the next
mode observation instant. 

Since the mode signal $\{r(k)\in\mathcal{M}\}_{k\in\mathbb{N}_{0}}$
is aperiodic and irreducible, mode sequences may start with any of
the possible modes indicated by the index set $\mathcal{M}=\{1,\ldots,M\}$.
Furthermore, it is possible to reach from any mode sequence to another
mode sequence in a finite number of mode observations. Hence, the
discrete-time Markov chain $\{s(i)\in\mathcal{S}\}_{i\in\mathbb{N}_{0}}$
is irreducible. In Lemma~\ref{invariantdistributionlemma} below,
we provide the invariant distribution for the countable-state discrete-time
Markov chain $\{s(i)\in\mathcal{S}\}_{i\in\mathbb{N}_{0}}$. Note
that the distribution $\phi:\mathcal{S}\to[0,1]:j\mapsto\phi_{j}$
is called \emph{invariant distribution} of the Markov chain $\{s(i)\in\mathcal{S}\}_{i\in\mathbb{N}_{0}}$
if $\phi_{j}=\sum_{i\in\mathcal{S}}\phi_{i}\rho_{i,j}$, $j\in\mathcal{S}$.
The invariant distribution for the case where $\mathcal{S}$ contains
only sequences of fixed length $T\in\mathbb{N}$ is provided in \citeasnoun{serfozo2009}.
In Lemma~\ref{invariantdistributionlemma}, we consider the more
general case where $\mathcal{S}$ may contain countably infinite number
of sequences of all possible lengths.

\begin{lem}\label{invariantdistributionlemma}Discrete-time Markov
chain $\{s(i)\in\mathcal{S}\}_{i\in\mathbb{N}_{0}}$ has invariant
distribution $\phi:\mathcal{S}\to[0,1]:q\mapsto\phi_{q}$ given by
\begin{align}
\phi_{q} & \triangleq\pi_{q_{1}}\mu_{|q|}\prod_{n=1}^{|q|-1}p_{q_{n},q_{n+1}},\quad q\in\mathcal{S},\label{eq:phiq}
\end{align}
 where $\pi:\mathcal{M}\to[0,1]$ and $p_{i,j}$, $i,j\in\mathcal{M}$,
respectively denote the invariant distribution and transition probabilities
of the finite-state Markov chain $\{r(k)\in\mathcal{M}\}_{k\in\mathbb{N}_{0}}$. 

\end{lem} 

\begin{proof} We prove this result by showing that $\phi_{\bar{q}}=\sum_{q\in\mathcal{S}}\phi_{q}\rho_{q,\bar{q}}$,
for all $\bar{q}\in\mathcal{S}$. First, by (\ref{eq:stransprob})
and (\ref{eq:phiq}) 
\begin{align}
\sum_{q\in\mathcal{S}}\phi_{q}\rho_{q,\bar{q}} & =\big(\sum_{q\in\mathcal{S}}\pi_{q_{1}}\mu_{|q|}\big(\prod_{n=1}^{|q|-1}p_{q_{n},q_{n+1}}\big)p_{q_{|q|},\bar{q}_{1}}\big)\nonumber \\
 & \quad\cdot\mu_{|\bar{q}|}\,\prod_{n=1}^{|\bar{q}|-1}p_{\bar{q}_{n},\bar{q}_{n+1}},\quad\bar{q}\in\mathcal{S}.\label{eq:phiqsum}
\end{align}
 Now let $\mathcal{S}_{\tau}\triangleq\{q\in\mathcal{S}\,:\,|q|=\tau\},$
$\tau\in\mathbb{N}$. Note that the set $\mathcal{S}_{\tau}$ contains
all mode sequences of length $\tau$. We rewrite the sum in (\ref{eq:phiqsum})
to obtain 
\begin{align}
 & \sum_{q\in\mathcal{S}}\pi_{q_{1}}\mu_{|q|}\big(\prod_{n=1}^{|q|-1}p_{q_{n},q_{n+1}}\big)p_{q_{|q|},\bar{q}_{1}}\nonumber \\
 & \,=\sum_{\tau\in\mathbb{N}}\mu_{\tau}\sum_{q\in\mathcal{S}_{\tau}}\pi_{q_{1}}\big(\prod_{n=1}^{\tau-1}p_{q_{n},q_{n+1}}\big)p_{q_{\tau},\bar{q}_{1}}\nonumber \\
 & \,=\sum_{\tau\in\mathbb{N}}\mu_{\tau}\sum_{q_{\tau}\in\mathcal{M}}\cdots\sum_{q_{1}\in\mathcal{M}}\pi_{q_{1}}\big(\prod_{n=1}^{\tau-1}p_{q_{n},q_{n+1}}\big)p_{q_{\tau},\bar{q}_{1}}.\label{eq:suminphiqsum}
\end{align}
 Note that since $\pi:\mathcal{M}\to[0,1]$ is the invariant distribution
of the finite-state Markov chain $\{r(k)\in\mathcal{M}\}_{k\in\mathbb{N}_{0}}$,
it follows that $\sum_{i\in\mathcal{M}}\pi_{i}p_{i,j}=\pi_{j}$, $i,j\in\mathcal{M}$.
Thus, we have $\sum_{q_{n}\in\mathcal{M}}\pi_{q_{n}}p_{q_{n},q_{n+1}}=\pi_{q_{n+1}}$,
$n\in\{1,\ldots,\tau-1\}$, and $\sum_{q_{\tau}\in\mathcal{M}}\pi_{q_{\tau}}p_{q_{\tau},\bar{q}_{1}}=\pi_{\bar{q}_{1}}$.
As a result, from (\ref{eq:suminphiqsum}) we obtain 
\begin{align}
\sum_{q\in\mathcal{S}}\pi_{q_{1}}\mu_{|q|}\big(\prod_{n=1}^{|q|-1}p_{q_{n},q_{n+1}}\big)p_{q_{|q|},\bar{q}_{1}} & =\sum_{\tau\in\mathbb{N}}\mu_{\tau}\pi_{\bar{q}_{1}}\nonumber \\
 & =\pi_{\bar{q}_{1}}.\label{eq:suminphiqsumfinalexpression}
\end{align}
 Finally, substituting (\ref{eq:suminphiqsumfinalexpression}) into
(\ref{eq:phiqsum}) yields 
\begin{align}
\sum_{q\in\mathcal{S}}\phi_{q}\rho_{q,\bar{q}} & =\pi_{\bar{q}_{1}}\mu_{|\bar{q}|}\,\prod_{n=1}^{|\bar{q}|-1}p_{\bar{q}_{n},\bar{q}_{n+1}}=\phi_{\bar{q}},\quad\bar{q}\in\mathcal{S},
\end{align}
which completes the proof. \end{proof} 

We have now established that the countable-state Markov chain $\{s(k)\in\mathcal{S}\}_{k\in\mathbb{N}_{0}}$
is irreducible and has the invariant distribution $\phi:\mathcal{S}\to[0,1]$
presented in Lemma~\ref{invariantdistributionlemma}. Note that the
strong law of large numbers (also called ergodic theorem; see \citeasnoun{norris2009}, \citeasnoun{serfozo2009}, \citeasnoun{durrett2010})
for discrete-time Markov chains states that $\mathbb{P}[\lim_{n\to\infty}\frac{1}{n}\sum_{k=0}^{n-1}\xi_{s(k)}=\sum_{i\in\mathcal{S}}\phi_{i}\xi_{i}]=1$,
for any $\xi_{i}\in\mathbb{R}$, $i\in\mathcal{S}$, such that $\sum_{i\in\mathcal{S}}\phi_{i}|\xi_{i}|<\infty$.
This result for the countable-state Markov chain $\{s(k)\in\mathcal{S}\}_{k\in\mathbb{N}_{0}}$
is crucial to obtain the main results of Section~\ref{sec:Sufficient-Conditions-for}
below. Specifically, in our stability analysis we utilize the ergodic
theorem for Markov chains. In the literature, for the stability analysis
of finite-mode \cite{bolzern2004almost} and infinite-mode \cite{li2012exponential}
discrete-time switched stochastic systems, researchers employed ergodic
theorem for the Markov chain that characterizes the mode signal. In
the next section, we use ergodic theorem for the Markov chain that
characterizes the sequence of mode values between consecutive mode
observation instants.

\section{Sufficient Conditions for Almost Sure Asymptotic Stabilization\label{sec:Sufficient-Conditions-for}}

In this section, we employ the results presented in Section~\ref{sec:SwitchedStochasticSection3}
to obtain sufficient conditions for almost sure asymptotic stabilization
of the closed-loop system (\ref{eq:control-system}) under the control
law (\ref{eq:control-law}). 

\begin{thm}\label{maintheorem}Consider the switched linear stochastic
system (\ref{eq:control-system}). If there exist matrices $\tilde{R}>0$,
$L_{i}\in\mathbb{R}^{m\times n},\, i\in\mathcal{M}$, and scalars
$\zeta_{i,j}\in(0,\infty)$, $i,j\in\mathcal{M}$, such that 
\begin{align}
 & \,0\geq(A_{i}\tilde{R}+B_{i}L_{j})^{\mathrm{T}}\tilde{R}^{-1}\nonumber \\
 & \,\,\quad\quad\cdot(A_{i}\tilde{R}+B_{i}L_{j})-\zeta_{i,j}\tilde{R},\quad i,j\in\mathcal{M},\label{eq:condp}\\
 & \sum_{\tau\in\mathbb{N}}\mu_{\tau}\sum_{l=1}^{\tau}\sum_{i,j\in\mathcal{M}}\pi_{i}p_{i,j}^{(l-1)}\ln\zeta_{j,i}<0,\label{eq:condzeta}
\end{align}
then the control law (\ref{eq:control-law}) with the feedback gain
matrix 
\begin{align}
K_{\sigma(k)} & =L_{\sigma(k)}\tilde{R}^{-1},\label{eq:controllawintheorem}
\end{align}
 guarantees that the zero solution $x(k)\equiv0$ of the closed-loop
system (\ref{eq:control-system}) and (\ref{eq:control-law}) is asymptotically
stable almost surely. \end{thm}

\begin{proof} First, we define $V(x)\triangleq x^{\mathrm{T}}Rx$,
where $R\triangleq\tilde{R}^{-1}$. It follows from (\ref{eq:control-system})
and (\ref{eq:control-law}) that for $k\in\mathbb{N}_{0}$, 
\begin{align}
V(x(k+1)) & =x^{\mathrm{T}}(k)(A_{r(k)}+B_{r(k)}K_{\sigma(k)})^{\mathrm{T}}R\nonumber \\
 & \,\,\,\,\,\,\,\,\cdot(A_{r(k)}+B_{r(k)}K_{\sigma(k)})x(k).\label{eq:vevolution}
\end{align}
We set $L_{j}=K_{j}R^{-1}$, $j\in\mathcal{M}$, and use (\ref{eq:condp})
and (\ref{eq:vevolution}) to obtain 
\begin{align}
V(x(k+1)) & \leq\zeta_{r(k),\sigma(k)}V(x(k))\leq\eta(k)V(x(0)),\label{eq:vinequality}
\end{align}
 for $k\in\mathbb{N}_{0}$, where $\eta(k)\triangleq\prod_{n=0}^{k}\zeta_{r(n),\sigma(n)}$,
$k\in\mathbb{N}$. We will first show that $\eta(k)\to0$ almost surely
as $k\to\infty$. Note that $\eta(k)>0$, $k\in\mathbb{N}_{0}$. Then,
it follows that 
\begin{align}
\ln\eta(k) & =\sum_{n=0}^{k}\ln\zeta_{r(n),\sigma(n)}.\label{eq:sumforlneta}
\end{align}
By using the definitions of stochastic processes $\{N(k)\in\mathbb{N}_{0}\}_{k\in\mathbb{N}_{0}}$
and $\{s(i)\in\mathcal{S}\}_{i\in\mathbb{N}_{0}}$, we obtain 
\begin{align}
\ln\eta(k) & =\sum_{n=0}^{t_{N(k)}-1}\ln\zeta_{r(n),\sigma(n)}+\sum_{n=t_{N(k)}}^{k}\ln\zeta_{r(n),\sigma(n)}\nonumber \\
 & =\sum_{i=0}^{N(k)-1}\xi_{s(i)}+\sum_{n=t_{N(k)}}^{k}\ln\zeta_{r(n),\sigma(n)},\label{eq:lneta}
\end{align}
 where $\xi_{q}\triangleq\sum_{n=1}^{|q|}\ln\zeta_{q_{n},q_{1}}$,
$q\in\mathcal{S}$. 

Next, in order to evaluate $\lim_{k\to\infty}\frac{1}{k}\ln\eta(k)$,
note that $\lim_{k\to\infty}\frac{1}{k}\sum_{n=t_{N(k)}}^{k}\ln\zeta_{r(n),\sigma(n)}=0.$
Consequently, 
\begin{align}
\lim_{k\to\infty}\frac{1}{k}\ln\eta(k) & =\lim_{k\to\infty}\frac{1}{k}\sum_{i=0}^{N(k)-1}\xi_{s(i)}\nonumber \\
 & =\lim_{k\to\infty}\frac{N(k)}{k}\frac{1}{N(k)}\sum_{i=0}^{N(k)-1}\xi_{s(i)}.
\end{align}
It follows from strong law of large numbers for renewal processes
(Section~\ref{sub:Discrete-Time-Renewal-Processes}) that $\lim_{k\to\infty}\frac{N(k)}{k}=\frac{1}{\hat{\tau}}$,
where $\hat{\tau}=\sum_{\tau\in\mathbb{N}}\tau\mu_{\tau}$. Furthermore,
by the ergodic theorem for countable-state Markov chains, it follows
that $\lim_{n\to\infty}\frac{1}{n}\sum_{i=0}^{n-1}\xi_{s(i)}=\sum_{q\in\mathcal{S}}\phi_{q}\xi_{q}$.
Using the invariant distribution $\phi:\mathcal{S}\to[0,1]$ given
by (\ref{eq:phiq}), we get 
\begin{align}
 & \lim_{k\to\infty}\frac{1}{k}\ln\eta(k)\nonumber \\
 & \quad=\frac{1}{\hat{\tau}}\sum_{q\in\mathcal{S}}\big(\pi_{q_{1}}\mu_{|q|}\prod_{n=1}^{|q|-1}p_{q_{n},q_{n+1}}\big)\sum_{m=1}^{|q|}\ln\zeta_{q_{m},q_{1}}.\label{eq:afterstronglaws}
\end{align}
Let $\mathcal{S}_{\tau}\triangleq\{q\in\mathcal{S}\,:\,|q|=\tau\},\,\tau\in\mathbb{N}$.
Note that $\mathcal{S}_{\tau}$ contains all mode sequences of length
$\tau$. It follows from (\ref{eq:afterstronglaws}) that 
\begin{align}
 & \lim_{k\to\infty}\frac{1}{k}\ln\eta(k)\nonumber \\
 & \quad=\frac{1}{\hat{\tau}}\sum_{\tau\in\mathbb{N}}\sum_{q\in\mathcal{S}_{\mathrm{\tau}}}\big(\pi_{q_{1}}\mu_{|q|}\prod_{n=1}^{|q|-1}p_{q_{n},q_{n+1}}\big)\sum_{m=1}^{|q|}\ln\zeta_{q_{m},q_{1}}\nonumber \\
 & \quad=\frac{1}{\hat{\tau}}\sum_{\tau\in\mathbb{N}}\mu_{\tau}\sum_{q\in\mathcal{S}_{\tau}}\pi_{q_{1}}(\prod_{n=1}^{\tau-1}p_{q_{n},q_{n+1}})\sum_{m=1}^{\tau}\ln\zeta_{q_{m},q_{1}}\nonumber \\
 & \quad=\frac{1}{\hat{\tau}}\sum_{\tau\in\mathbb{N}}\mu_{\tau}\sum_{m=1}^{\tau}\sum_{q\in\mathcal{S}_{\tau}}\pi_{q_{1}}(\prod_{n=1}^{\tau-1}p_{q_{n},q_{n+1}})\ln\zeta_{q_{m},q_{1}}.\label{eq:afterstdef}
\end{align}
Furthermore, let $\mathcal{S}_{\tau,l}^{i,j}\triangleq\{q\in\mathcal{S}_{\tau}:q_{1}=i,q_{l}=j\}$,
$i,j\in\mathcal{M}$, $l\in\{1,2,\ldots,\tau-1\}$. The set $\mathcal{S}_{\tau,l}^{i,j}$
contains all mode sequences of length $\tau$ that have $i\in\mathcal{M}$
and $j\in\mathcal{M}$ as the $1$st and the $l$th elements, respectively.
We use (\ref{eq:pton}) to obtain 
\begin{align}
 & \sum_{q\in\mathcal{S}_{\tau}}\pi_{q_{1}}(\prod_{n=1}^{\tau-1}p_{q_{n},q_{n+1}})\ln\zeta_{q_{l},q_{1}}\nonumber \\
 & \quad=\sum_{i,j\in\mathcal{M}}\sum_{q\in\mathcal{S}_{\tau,l}^{i,j}}\pi_{q_{1}}(\prod_{n=1}^{\tau-1}p_{q_{n},q_{n+1}})\ln\zeta_{q_{l},q_{1}}\nonumber \\
 & \quad=\sum_{i,j\in\mathcal{M}}\pi_{i}(\ln\zeta_{j,i})\sum_{q\in\mathcal{S}_{\tau,l}^{i,j}}(\prod_{n=1}^{\tau-1}p_{q_{n},q_{n+1}})\nonumber \\
 & \quad=\sum_{i,j\in\mathcal{M}}\pi_{i}(\ln\zeta_{j,i})p_{i,j}^{(l-1)}.\label{eq:usestijdef}
\end{align}
Substituting (\ref{eq:usestijdef}) into (\ref{eq:afterstdef}) yields
\begin{align}
\lim_{k\to\infty}\frac{1}{k}\ln\eta(k) & =\frac{1}{\hat{\tau}}\sum_{\tau\in\mathbb{N}}\mu_{\tau}\sum_{l=1}^{\tau}\sum_{i,j\in\mathcal{M}}\pi_{i}p_{i,j}^{(l-1)}\ln\zeta_{j,i}.
\end{align}
Now, since $\hat{\tau}=\sum_{\tau\in\mathbb{N}}\tau\mu_{\tau}<\infty$,
as a result of (\ref{eq:condzeta}), we have $\lim_{k\to\infty}\frac{1}{k}\ln\eta(k)<0$.
Thus, $\lim_{k\to\infty}\ln\eta(k)=-\infty$ almost surely; furthermore,
$\mathbb{P}[\lim_{k\to\infty}\eta(k)=0]=1.$ In the following, we
first show that the zero solution is \emph{almost surely stable}.
To this end first note that for all $\epsilon>0$, $\lim_{n\to\infty}\mathbb{P}[\sup_{k\geq n}\eta(k)>\epsilon^{2}]=0,$
which implies that for all $\epsilon>0$ and $\rho>0$, there exists
a positive integer $N(\epsilon,\rho)$ such that $\mathbb{P}[\sup_{k\geq n}\eta(k)>\epsilon^{2}]<\rho$
for $n\geq N(\epsilon,\rho)$. Equivalently, 
\begin{align}
\mathbb{P}[\sup_{k\geq n}\eta(k-1)>\epsilon^{2}]<\rho,\quad n\geq N(\epsilon,\rho)+1.\label{eq:epsilon-rho-inequlity-k-minus-one}
\end{align}
By the definition of $V(\cdot)$ and (\ref{eq:vinequality}), we obtain
$\eta(k-1)\geq\frac{V(x(k))}{V(x(0))}\geq\frac{\lambda_{\min}(R)}{\lambda_{\max}(R)}\frac{\|x(k)\|^{2}}{\|x(0)\|^{2}}$
for all $k\in\mathbb{N}$. Hence, it follows from (\ref{eq:epsilon-rho-inequlity-k-minus-one})
that, for all $\epsilon>0$ and $\rho>0$, there exists a positive
integer $N(\epsilon,\rho)$ such that 
\begin{align}
 & \mathbb{P}[\sup_{k\geq n}\|x(k)\|>\epsilon\sqrt{\frac{\lambda_{\max}(R)}{\lambda_{\min}(R)}}\|x(0)\|]\nonumber \\
 & \quad=\mathbb{P}[\sup_{k\geq n}\|x(k)\|^{2}>\epsilon^{2}\frac{\lambda_{\max}(R)}{\lambda_{\min}(R)}\|x(0)\|^{2}]\nonumber \\
 & \quad=\mathbb{P}[\sup_{k\geq n}\frac{\lambda_{\min}(R)}{\lambda_{\max}(R)}\frac{\|x(k)\|^{2}}{\|x(0)\|^{2}}>\epsilon^{2}]\nonumber \\
 & \quad\leq\mathbb{P}[\sup_{k\geq n}\eta(k-1)>\epsilon^{2}]<\rho,\quad n\geq N(\epsilon,\rho)+1.
\end{align}
 Let $\delta_{1}\triangleq\sqrt{\frac{\lambda_{\min}(R)}{\lambda_{\max}(R)}}$.
If $\|x(0)\|\leq\delta_{1}$, then 
\begin{align}
 & \mathbb{P}[\sup_{k\geq n}\|x(k)\|>\epsilon]\nonumber \\
 & \quad\leq\mathbb{P}[\sup_{k\geq n}\|x(k)\|>\epsilon\sqrt{\frac{\lambda_{\max}(R)}{\lambda_{\min}(R)}}\|x(0)\|]\nonumber \\
 & \quad<\rho,\quad n\geq N(\epsilon,\rho)+1.\label{eq:epsilon-result-part1}
\end{align}
Now let $\bar{\zeta}\triangleq\max\{1,\max_{i,j\in\mathcal{M}}\zeta_{i,j}\}$.
It follows from (\ref{eq:vinequality}) that $V(x(k))\leq\bar{\zeta}^{k-1}V(x(0))\leq\bar{\zeta}^{N(\epsilon,\rho)-1}V(x(0))$
for all $k\in\{0,1,\ldots,N(\epsilon,\rho)\}$. Therefore, $\|x(k)\|^{2}\leq\bar{\zeta}^{N(\epsilon,\rho)-1}\frac{\lambda_{\max}(R)}{\lambda_{\min}(R)}\|x(0)\|^{2}$,
and hence, we have $\|x(k)\|\leq\sqrt{\bar{\zeta}^{N(\epsilon,\rho)-1}\frac{\lambda_{\max}(R)}{\lambda_{\min}(R)}}\|x(0)\|,$
for all $k\in\{0,1,\ldots,N(\epsilon,\rho)\}$. Furthermore, let $\delta_{2}\triangleq\epsilon\sqrt{\bar{\zeta}^{-N(\epsilon,\rho)+1}\frac{\lambda_{\min}(R)}{\lambda_{\max}(R)}}$.
Consequently, if $\|x(0)\|\leq\delta_{2}$, then $\|x(k)\|\leq\epsilon$,
$k\in\{0,1,\ldots,N(\epsilon,\rho)\}$, which implies 
\begin{eqnarray}
\mathbb{P}[\max_{k\in\{0,1,\ldots,N(\epsilon,\rho)\}}\|x(k)\|>\epsilon] & = & 0.\label{eq:epsilon-result-part2}
\end{eqnarray}
It follows from (\ref{eq:epsilon-result-part1}) and (\ref{eq:epsilon-result-part2})
that for all $\epsilon>0$, $\rho>0$, 
\begin{align}
\mathbb{P}[\sup_{k\in\mathbb{N}_{0}}\|x(k)\|>\epsilon] & =\mathbb{P}[\{\max_{k\in\{0,1,\ldots,N(\epsilon,\rho)\}}\|x(k)\|>\epsilon\}\nonumber \\
 & \quad\quad\cup\,\{\sup_{k\geq N(\epsilon,\rho)+1}\|x(k)\|>\epsilon\}]\nonumber \\
 & \leq\mathbb{P}[\max_{k\in\{0,1,\ldots,N(\epsilon,\rho)\}}\|x(k)\|>\epsilon]\nonumber \\
 & \quad+\mathbb{P}[\sup_{k\geq N(\epsilon,\rho)+1}\|x(k)\|>\epsilon]\nonumber \\
 & <\rho,
\end{align}
 whenever $\|x(0)\|<\delta\triangleq\min(\delta_{1},\delta_{2})$,
which implies almost sure stability. As a final step of proving almost
sure asymptotic stability of the zero solution, we now show (\ref{eq:definition-convergence}).
First, note that by (\ref{eq:vinequality}), we have $V(x(k+1))\leq\eta(k)V(x(0))$,
$k\in\mathbb{N}$. Now, since $\mathbb{P}[\lim_{k\to\infty}\eta(k)=0]=1$,
it follows that $\mathbb{P}[\lim_{k\to\infty}V(x(k))=0]=1$, which
implies (\ref{eq:definition-convergence}), and hence the zero solution
of the closed-loop system (\ref{eq:control-system}), (\ref{eq:control-law})
is asymptotically stable almost surely. \end{proof}

Theorem~\ref{maintheorem} provides sufficient conditions for almost
sure asymptotic stability of the closed-loop system (\ref{eq:control-system})
and (\ref{eq:control-law}). Conditions (\ref{eq:condp}) and (\ref{eq:condzeta})
of Theorem~\ref{maintheorem} indicate dependence of stabilization
performance on subsystem dynamics, mode transition probabilities,
and random mode observations. The effect of mode transitions on the
stabilization is reflected in (\ref{eq:condp}) through the limiting
distribution $\pi:\mathcal{M}\to[0,1]$ as well as $l$-step transition
probabilities $p_{i,j}^{(l)}$, $i,j\in\mathcal{M}$. Furthermore,
the effect of random mode observations is indicated in condition (\ref{eq:condp})
by $\mu:\mathbb{N}\to[0,1]$, which represents the distribution of
the lengths of intervals between consecutive mode observation instants. 

\begin{rem} \label{conservatism-remark} We investigate the stability
of the closed-loop system through the Lyapunov-like function $V(x)\triangleq x^{\mathrm{T}}Rx$
with $R=\tilde{R}^{-1}$, where $\tilde{R}$ is a positive-definite
matrix that satisfy (\ref{eq:condp}). The scalar $\zeta_{i,j}\in(0,\infty)$
in (\ref{eq:condp}) characterizes an upper bound on the growth of
the Lyapunov-like function, when the switched system evolves according
to dynamics of the $i$th subsystem and the $j$th feedback gain.
Note that if $\zeta_{i,j}\in(0,1)$ for all $i,j\in\mathcal{M}$,
it is guaranteed that the Lyapunov-like function will decrease at
each time step. However, we do not require $\zeta_{i,j}\in(0,1)$
for all $i,j\in\mathcal{M}$. There may be pairs $i,j\in\mathcal{M}$
such that $\zeta_{i,j}>1$, hence Lyapunov-like function $V(\cdot)$
may grow when $i$th subsystem and the $j$th feedback gain is active.
As long as $\zeta_{i,j}$, $i,j\in\mathcal{M}$, satisfy (\ref{eq:condzeta})
the Lyapunov-like is guaranteed to converge to zero in the long-run
(even if it may grow at certain instants). Note that even though the
conditions (\ref{eq:condp}), (\ref{eq:condzeta}) allow unstable
subsystem-feedback gain pairs, some conservativeness may still arise
due the characterization with single Lyapunov-like function. This
conservatism may be reduced with an alternative approach with multiple
Lyapunov-like functions assigned for each subsystem-feedback gain
pairs. \end{rem}

\begin{rem}In order to verify conditions (\ref{eq:condp}) and (\ref{eq:condzeta})
of Theorem~\ref{maintheorem}, we take an approach similar to the
one presented in \citeasnoun{cetinkaya2013a}. Specifically, we use
Schur complements (see \citeasnoun{bernstein2009matrix}) to transform
condition (\ref{eq:condp}) into the matrix inequalities 
\begin{align}
0\leq & \left[\begin{array}{cc}
\zeta_{i,j}\tilde{R} & \hat{A}_{i,j}^{\mathrm{T}}\\
\hat{A}_{i,j} & \tilde{R}
\end{array}\right],\quad i,j\in\mathcal{M},\label{eq:lmi}
\end{align}
 where $\hat{A}_{i,j}\triangleq(A_{i}\tilde{R}+B_{i}L_{j}),\, i,j\in\mathcal{M}$.
Note that the inequalities (\ref{eq:lmi}) are linear in $\tilde{R}$
and $L_{i}$, $i\in\mathcal{M}$. In our numerical method, we iterate
over a set of the values of $\zeta_{i,j}$, $i,j\in\mathcal{M}$,
that satisfy (\ref{eq:condzeta}) and at each iteration we look for
feasible solutions to the linear matrix inequalities (\ref{eq:lmi}).
In Section~\ref{sec:Illustrative-Numerical-Example} below, we employ
this method and find values for matrices $\tilde{R}\in\mathbb{R}^{n\times n},L_{i}\in\mathbb{R}^{m\times n},i\in\mathcal{M}$,
and scalars $\zeta_{i,j}\in(0,\infty),i,j\in\mathcal{M}$, that satisfy
(\ref{eq:condp}), (\ref{eq:condzeta}) for a given discrete-time
switched linear system. It is important to note that the scalars $\zeta_{i,j}\in(0,\infty),i,j\in\mathcal{M}$,
that satisfy (\ref{eq:condzeta}) form an unbounded set. Note that
this set is smaller than the entire nonnegative orthant in $\mathbb{R}^{M^{2}}$.
However, we still need to reduce the search space of $\zeta_{i,j},i,j\in\mathcal{M}$.
To this end, first note that it is harder to find feasible solutions
to linear matrix inequalities given by (\ref{eq:lmi}) when the scalars
$\zeta_{i,j},i,j\in\mathcal{M},$ are close to zero. Note also that
if there exist a feasible solution to (\ref{eq:lmi}) for certain
values of $\zeta_{i,j},i,j\in\mathcal{M},$ then it is guaranteed
that feasible solutions to (\ref{eq:lmi}) exist also for \emph{larger
}values of $\zeta_{i,j},i,j\in\mathcal{M}$. Therefore, we can restrict
our search space and iterate over large values of $\zeta_{i,j},i,j\in\mathcal{M},$
that satisfy (\ref{eq:condzeta}), and check feasible solutions to
(\ref{eq:lmi}). Specifically, we only iterate over $\zeta_{i,j},i,j\in\mathcal{M},$
that is close to the search space's boundary identified by $\sum_{\tau\in\mathbb{N}}\mu_{\tau}\sum_{l=1}^{\tau}\sum_{i,j\in\mathcal{M}}\pi_{i}p_{i,j}^{(l-1)}\ln\zeta_{j,i}=0$.
Now note that in order for (\ref{eq:condzeta}) to be satisfied, there
must exist at least a pair $i,j\in\mathcal{M}$ such that $\zeta_{i,j}<1$.
Since the scalar $\zeta_{i,j}$ represents the stability/instability
margin for the dynamics characterized by the $i$th subsystem and
the $j$th feedback gain, we expect $\zeta_{i,i}<1$ for stabilizable
modes $i\in\mathcal{M}$. This further reduces the search space for
our numerical method. 

\end{rem}

\begin{rem}Note that conditions (\ref{eq:condp}) and (\ref{eq:condzeta})
presented in Theorem~\ref{maintheorem} can also be used for determining
almost sure asymptotic stability of the switched stochastic control
system (\ref{eq:control-system}), (\ref{eq:control-law}) with periodically
observed mode information. The renewal process characterization presented
in this paper in fact encompasses periodic mode observations (explored
previously in \citeasnoun{cetinkayaacc2012} and \citeasnoun{cetinkaya2013a})
as a special case. Specifically, suppose that the mode observation
instants are given by $t_{i}=iT$, $i\in\mathbb{N}_{0}$, where $T\in\mathbb{N}$
denotes the mode observation period. Our present framework allows
us to characterize periodic mode observations by setting the distribution
$\mu:\mathbb{N}\to[0,1]$ such that $\mu_{T}=1$ and $\mu_{\tau}=0$,
$\tau\neq T$. Note that condition (\ref{eq:condzeta}) of Theorem~\ref{maintheorem}
for this case reduces to $\sum_{l=1}^{T}\sum_{i,j\in\mathcal{M}}\pi_{i}p_{i,j}^{(l-1)}\ln\zeta_{j,i}<0$.
Furthermore, if the controller has perfect mode information at all
time instants ($T=1$, hence $\sigma(k)=r(k)$, $k\in\mathbb{N}_{0}$),
condition (\ref{eq:condzeta}) takes even a simpler form given by
the inequality $\sum_{i\in\mathcal{M}}\pi_{i}\ln\zeta_{i,i}<0$. \end{rem} 

\begin{figure}[t]
\begin{center}\includegraphics[width=0.9\columnwidth]{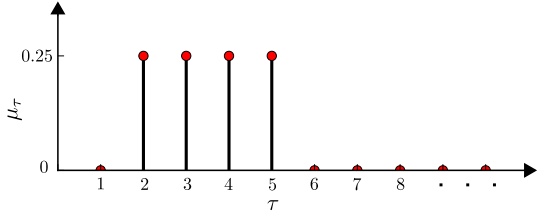}\end{center}
\vskip -5pt\protect\caption{Uniform distribution given by (\ref{eq:uniformmu}) with $\tau_{L}=2$
and $\tau_{H}=5$ for the length of intervals between consecutive
mode observation instants }
\label{Flo:muunif}
\end{figure}

\begin{rem} \label{remarkforexample2} Condition (\ref{eq:condzeta})
of Theorem~\ref{maintheorem} has a simpler form also for the case
where the length of intervals between consecutive mode observation
instants are uniformly distributed over the set $\{\tau_{\mathrm{L}},\tau_{\mathrm{L}}+1,\ldots,\tau_{\mathrm{H}}\}$
with $\tau_{\mathrm{L}},\tau_{\mathrm{H}}\in\mathbb{N}$ such that
$\tau_{\mathrm{L}}\leq\tau_{\mathrm{H}}$. In this case the distribution
$\mu:\mathbb{N}\to[0,1]$ is given by 
\begin{align}
\mu_{\tau} & \triangleq\begin{cases}
\frac{1}{\tau_{\mathrm{H}}-\tau_{\mathrm{L}}+1}, & \quad\mathrm{if}\,\,\,\tau\in\{\tau_{\mathrm{L}},\tau_{\mathrm{L}}+1,\ldots,\tau_{\mathrm{H}}\},\\
0, & \quad\mathrm{otherwise}.
\end{cases}\label{eq:uniformmu}
\end{align}
Figure~\ref{Flo:muunif} shows the distribution (\ref{eq:uniformmu})
for an example case with $\tau_{\mathrm{L}}=2$ and $\tau_{\mathrm{H}}=5$. 

With (\ref{eq:uniformmu}), condition (\ref{eq:condzeta}) of Theorem~\ref{maintheorem}
reduces to the inequality $\sum_{\tau=\tau_{\mathrm{L}}}^{\tau_{\mathrm{H}}}\sum_{l=1}^{\tau}\sum_{i,j\in\mathcal{M}}\pi_{i}p_{i,j}^{(l-1)}\ln\zeta_{j,i}<0.$ 

\end{rem}

\begin{rem} \label{remarkfortheexample} Note that our probabilistic
characterization of mode observation instants also allows us to explore
the feedback control problem under missing mode samples. Specifically,
consider the case where the mode is sampled at all time instants;
however, some of the mode samples are lost during communication between
mode sampling mechanism and the controller. Suppose that the controller
receives a sampled mode data at each time step $k\in\mathbb{N}$ with
probability $\theta\in(0,1)$. In other words, the mode data is lost
with probability $1-\theta$. We investigate this problem by setting
\begin{align}
\mu_{\tau} & \triangleq(1-\theta)^{\tau-1}\theta,\quad\tau\in\mathbb{N}.\label{eq:muformissing}
\end{align}
Figure~\ref{Flo:muwiththeta} shows the distribution (\ref{eq:muformissing})
with $\theta=0.3$. 

It turns out that for $\mu_{\tau}:\mathbb{N}\to[0,1]$ given by (\ref{eq:muformissing}),
the left-hand side of condition (\ref{eq:condzeta}) has a closed-form
expression. Note that by changing the order of summations and using
(\ref{eq:muformissing}), we can rewrite the left-hand side of (\ref{eq:condzeta})
as 
\begin{align}
 & \sum_{\tau\in\mathbb{N}}\mu_{\tau}\sum_{l=1}^{\tau}\sum_{i,j\in\mathcal{M}}\pi_{i}p_{i,j}^{(l-1)}\ln\zeta_{j,i}\nonumber \\
 & \,\,=\sum_{i,j\in\mathcal{M}}\pi_{i}(\ln\zeta_{j,i})\sum_{\tau\in\mathbb{N}}\mu_{\tau}\sum_{l=1}^{\tau}p_{i,j}^{(l-1)}\nonumber \\
 & \,\,=\sum_{i,j\in\mathcal{M}}\pi_{i}(\ln\zeta_{j,i})\sum_{l=1}^{\infty}p_{i,j}^{(l-1)}\sum_{\tau=l}^{\infty}\mu_{\tau}\nonumber \\
 & \,\,=\sum_{i,j\in\mathcal{M}}\pi_{i}(\ln\zeta_{j,i})\sum_{l=1}^{\infty}p_{i,j}^{(l-1)}(1-\sum_{\tau=1}^{l-1}\mu_{\tau})\nonumber \\
 & \,\,=\sum_{i,j\in\mathcal{M}}\pi_{i}(\ln\zeta_{j,i})\sum_{l=1}^{\infty}p_{i,j}^{(l-1)}\big(1-\sum_{\tau=1}^{l-1}(1-\theta)^{\tau-1}\theta\big).
\end{align}
Note that $\big(1-\sum_{\tau=1}^{l-1}(1-\theta)^{\tau-1}\theta\big)=\big(1-\theta\frac{1-(1-\theta)^{l-1}}{1-(1-\theta)}\big)=(1-\theta)^{l-1}$.
Therefore, 
\begin{align}
 & \sum_{\tau\in\mathbb{N}}\mu_{\tau}\sum_{l=1}^{\tau}\sum_{i,j\in\mathcal{M}}\pi_{i}p_{i,j}^{(l-1)}\ln\zeta_{j,i}\nonumber \\
 & \quad=\sum_{i,j\in\mathcal{M}}\pi_{i}(\ln\zeta_{j,i})\sum_{l=1}^{\infty}p_{i,j}^{(l-1)}(1-\theta)^{l-1}.\label{eq:beforematrixdef}
\end{align}
Let $Z\triangleq\sum_{l=1}^{\infty}P^{l-1}(1-\theta)^{l-1},$ where
$P\in\mathbb{R}^{M\times M}$ denotes the transition probability matrix
for the mode signal $\{r(k)\in\mathcal{M}\}_{k\in\mathbb{N}_{0}}$.
Note that the infinite sum in the definition of $Z$ converges, because
the eigenvalues of the matrix $(1-\theta)P$ are strictly inside the
unit circle of the complex plane. By using the formula for geometric
series of matrices \cite{bernstein2009matrix}, we obtain $Z=\big(I-(1-\theta)P)^{-1}$.
Furthermore, it follows from (\ref{eq:beforematrixdef}) that $\sum_{\tau\in\mathbb{N}}\mu_{\tau}\sum_{l=1}^{\tau}\sum_{i,j\in\mathcal{M}}\pi_{i}p_{i,j}^{(l-1)}\ln\zeta_{j,i}=\sum_{i,j\in\mathcal{M}}\pi_{i}(\ln\zeta_{j,i})z_{i,j}$,
and therefore, when $\mu:\mathbb{N}\to[0,1]$ is given by (\ref{eq:muformissing}),
condition (\ref{eq:condzeta}) takes the form $\sum_{i,j\in\mathcal{M}}\pi_{i}(\ln\zeta_{j,i})z_{i,j}<0$,
where $z_{i,j}$ is the $(i,j)$th entry of the matrix $Z$. 

\begin{figure}[t]
\begin{center}\includegraphics[width=0.9\columnwidth]{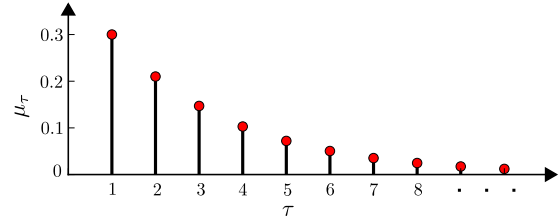}\end{center}
\vskip -5pt\protect\caption{Distribution given by (\ref{eq:muformissing}) with $\theta=0.3$
for the length of intervals between consecutive mode observation instants }
\label{Flo:muwiththeta}
\end{figure}

\end{rem}

\begin{rem}Note that in order to check condition (\ref{eq:condzeta})
of Theorem~\ref{maintheorem}, one needs to have perfect information
regarding the distribution $\mu:\mathbb{N}\to[0,1]$, according to
which the lengths of intervals between consecutive mode observation
instants are distributed. In Theorem~\ref{relaxtautheorem} below,
we present alternative sufficient stabilization conditions, which
do not require exact knowledge of $\mu:\mathbb{N}\to[0,1]$. Specifically,
we consider the case where the \emph{mode observation instants} $t_{i}$,
$i\in\mathbb{N}_{0}$, satisfy 
\begin{align}
\mathbb{P}[t_{i+1}-t_{i}\leq\bar{\tau}] & =1,\quad i\in\mathbb{N}_{0},\label{eq:newcondformondeobservationinstants}
\end{align}
 where $\bar{\tau}\in\mathbb{N}$ is a known constant. In this case
time instants of consecutive mode observations are assumed to be at
most $\bar{\tau}\in\mathbb{N}$ steps apart. In other words, if (\ref{eq:newcondformondeobservationinstants})
is satisfied, it is guaranteed that the length of intervals between
consecutive mode observation instants cannot be larger than $\bar{\tau}\in\mathbb{N}$.
It is important to note that (\ref{eq:newcondformondeobservationinstants})
characterizes a requirement on the intervals between \emph{mode observation
instants} and it is not related to mode switches. 

\begin{thm}\label{relaxtautheorem}Consider the switched linear stochastic
system (\ref{eq:control-system}). Suppose that the mode-transition
probability matrix $P\in\mathbb{R}^{M\times M}$ possesses only positive
real eigenvalues. If there exist matrices $\tilde{R}>0$, $L_{i}\in\mathbb{R}^{m\times n},\, i\in\mathcal{M}$,
and scalars $\bar{\tau}\in\mathbb{N}$, $\zeta_{i,j}\in(0,\infty)$,
$i,j\in\mathcal{M}$, such that (\ref{eq:condp}), (\ref{eq:newcondformondeobservationinstants}),
\begin{align}
 & \,\,\,\,0\leq\zeta_{j,i}-\zeta_{i,i},\quad i,j\in\mathcal{M},\label{eq:zetarelaxingcondition}\\
 & \sum_{l=1}^{\bar{\tau}}\sum_{i,j\in\mathcal{M}}\pi_{i}p_{i,j}^{(l-1)}\ln\zeta_{j,i}<0,\label{eq:newzetacond}
\end{align}
hold, then the control law (\ref{eq:control-law}) with the feedback
gain matrix (\ref{eq:controllawintheorem}) guarantees that the zero
solution $x(k)\equiv0$ of the closed-loop system is asymptotically
stable almost surely.\end{thm}

\begin{proof}The mode signal $\{r(k)\in\mathcal{M}\}_{k\in\mathbb{N}_{0}}$
is an irreducible and aperiodic Markov chain; therefore, the invariant
distribution $\pi:\mathcal{M}\to[0,1]$ is also the limiting distribution
\cite{norris2009}. Thus, for all $i,j\in\mathcal{M}$ and $k\in\mathbb{N}_{0}$,
\begin{align}
\lim_{l\to\infty}p_{i,j}^{(l)} & =\lim_{l\to\infty}\mathbb{P}[r(k+l)=j|r(k)=i]=\pi_{j}.
\end{align}
Now, let $p_{i}^{(l)}\in[0,1]^{1\times M}$,$i\in\mathcal{M}$, denote
the row vector with the $j$th element given by the $l$-step transition
probability $p_{i,j}^{(l)}$. Note that $p_{i}^{(\cdot)}$ is the
unique solution of the difference equation 
\begin{align}
p_{i}^{(l+1)} & =p_{i}^{(l)}P,\quad l\in\mathbb{N}_{0},\label{eq:difeq}
\end{align}
 with the initial condition $p_{i,i}^{(0)}=1$ and $p_{i,j}^{(0)}=0$,
$i\neq j$, $j\in\mathcal{M}$. Since all the eigenvalues of the mode-transition
probability matrix $P\in\mathbb{R}^{M\times M}$ are positive real
numbers, the solution $p_{i}^{(\cdot)}$ of the difference equation
(\ref{eq:difeq}) does not comprise any oscillatory components, and
$l$-step transition probabilities $p_{i,j}^{(l)}$, $i,j\in\mathcal{M}$,
converge towards their limiting values \emph{monotonically}, that
is,  
\begin{align}
p_{i,i}^{(l+1)} & \leq p_{i,i}^{(l)},\quad i\in\mathcal{M},\,\, l\in\mathbb{N}_{0},\label{eq:lstepineq1}\\
p_{i,j}^{(l+1)} & \geq p_{i,j}^{(l)},\quad i\neq j,\,\,\, i,j\in\mathcal{M},\,\,\, l\in\mathbb{N}_{0}.\label{eq:lstepineq2}
\end{align}
Now note that for all $i,j\in\mathcal{M}$, and $\tau\in\mathbb{N}$,
\begin{align}
\frac{1}{\tau}\sum_{l=1}^{\tau}p_{i,j}^{(l-1)} & =\frac{1}{\tau+1}\big(\sum_{l=1}^{\tau}p_{i,j}^{(l-1)}+\frac{1}{\tau}\sum_{l=1}^{\tau}p_{i,j}^{(l-1)}\big).\label{eq:scaledpij}
\end{align}
By (\ref{eq:lstepineq2}), we have $p_{i,j}^{(l-1)}\leq p_{i,j}^{\tau}$,
$l\in\{1,2,\ldots,\tau\}$, $i,j\in\mathcal{M}$, $i\neq j$. Hence,
it follows from (\ref{eq:scaledpij}) that 
\begin{align}
\frac{1}{\tau}\sum_{l=1}^{\tau}p_{i,j}^{(l-1)} & \leq\frac{1}{\tau+1}\big(\sum_{l=1}^{\tau}p_{i,j}^{(l-1)}+\frac{1}{\tau}\sum_{l=1}^{\tau}p_{i,j}^{(\tau)}\big)\nonumber \\
 & =\frac{1}{\tau+1}\big(\sum_{l=1}^{\tau}p_{i,j}^{(l-1)}+p_{i,j}^{(\tau)}\big)\nonumber \\
 & =\frac{1}{\tau+1}\sum_{l=1}^{\tau+1}p_{i,j}^{(l-1)},\quad\tau\in\mathbb{N},\,\,\, i\neq j.\label{eq:scaledpijineq}
\end{align}
As a consequence, for all $\tau\leq\bar{\tau}$ it follows that 
\begin{align}
\frac{1}{\tau}\sum_{l=1}^{\tau}p_{i,j}^{(l-1)} & \leq\frac{1}{\bar{\tau}}\sum_{l=1}^{\bar{\tau}}p_{i,j}^{(l-1)},\quad i\neq j,\,\,\, i,j\in\mathcal{M}.\label{eq:lsumineq2}
\end{align}
Next, we show that (\ref{eq:newcondformondeobservationinstants})--(\ref{eq:newzetacond})
together with (\ref{eq:lsumineq2}) imply (\ref{eq:condzeta}). First,
let $\kappa_{\tau,\bar{\tau}}^{i,j}\triangleq\frac{1}{\tau}\sum_{l=1}^{\tau}p_{i,j}^{(l-1)}-\frac{1}{\bar{\tau}}\sum_{l=1}^{\bar{\tau}}p_{i,j}^{(l-1)}$,
$i,j\in\mathcal{M}$. It follows that 
\begin{align}
 & \frac{1}{\tau}\sum_{l=1}^{\tau}\sum_{i,j\in\mathcal{M}}\pi_{i}p_{i,j}^{(l-1)}\ln\zeta_{j,i}\nonumber \\
 & \quad=\sum_{i,j\in\mathcal{M}}\pi_{i}\ln\zeta_{j,i}\frac{1}{\tau}\sum_{l=1}^{\tau}p_{i,j}^{(l-1)}\nonumber \\
 & \quad=\sum_{i,j\in\mathcal{M}}\pi_{i}(\ln\zeta_{j,i})\kappa_{\tau,\bar{\tau}}^{i,j}+\frac{1}{\bar{\tau}}\sum_{l=1}^{\bar{\tau}}\sum_{i,j\in\mathcal{M}}\pi_{i}p_{i,j}^{(l-1)}\ln\zeta_{j,i}\nonumber \\
 & \quad=\sum_{i\in\mathcal{M}}\pi_{i}(\ln\zeta_{i,i})\kappa_{\tau,\bar{\tau}}^{i,i}+\sum_{i\in\mathcal{M}}\sum_{j\in\mathcal{M},j\neq i}\pi_{i}(\ln\zeta_{j,i})\kappa_{\tau,\bar{\tau}}^{i,j}\nonumber \\
 & \quad\quad+\frac{1}{\bar{\tau}}\sum_{l=1}^{\bar{\tau}}\sum_{i,j\in\mathcal{M}}\pi_{i}p_{i,j}^{(l-1)}\ln\zeta_{j,i}.\label{eq:kappaineq}
\end{align}
Note that by (\ref{eq:lsumineq2}), we have $\kappa_{\tau,\bar{\tau}}^{i,j}\leq0$,
$\tau\leq\bar{\tau}$, $i\neq j$. It follows from (\ref{eq:zetarelaxingcondition})
that, for $\tau\leq\bar{\tau}$, 
\begin{align}
(\ln\zeta_{j,i})\kappa_{\tau,\bar{\tau}}^{i,j} & \leq(\ln\zeta_{i,i})\kappa_{\tau,\bar{\tau}}^{i,j},\quad i\neq j,\,\, i,j\in\mathcal{M}.\label{eq:logineq}
\end{align}
Now, since $\sum_{j\in\mathcal{M}}p_{i,j}^{(l)}=1$, $l\in\mathbb{N}_{0}$,
$i\in\mathcal{M}$, we have 
\begin{align}
\sum_{j\in\mathcal{M}}\kappa_{\tau,\bar{\tau}}^{i,j} & =\sum_{j\in\mathcal{M}}\frac{1}{\tau}\sum_{l=1}^{\tau}p_{i,j}^{(l-1)}-\sum_{j\in\mathcal{M}}\frac{1}{\bar{\tau}}\sum_{l=1}^{\bar{\tau}}p_{i,j}^{(l-1)}\nonumber \\
 & =\frac{1}{\tau}\sum_{l=1}^{\tau}\sum_{j\in\mathcal{M}}p_{i,j}^{(l-1)}-\frac{1}{\bar{\tau}}\sum_{l=1}^{\bar{\tau}}\sum_{j\in\mathcal{M}}p_{i,j}^{(l-1)}\nonumber \\
 & =\frac{\tau}{\tau}-\frac{\bar{\tau}}{\bar{\tau}}=0,\quad i\in\mathcal{M}.\label{eq:kappasum}
\end{align}
We use (\ref{eq:kappaineq})--(\ref{eq:kappasum}) to obtain 
\begin{align}
 & \frac{1}{\tau}\sum_{l=1}^{\tau}\sum_{i,j\in\mathcal{M}}\pi_{i}p_{i,j}^{(l)}\ln\zeta_{j,i}\nonumber \\
 & \quad\leq\sum_{i\in\mathcal{M}}\pi_{i}(\ln\zeta_{i,i})\kappa_{\tau,\bar{\tau}}^{i,j}+\sum_{i\in\mathcal{M}}\sum_{j\in\mathcal{M},j\neq i}\pi_{i}(\ln\zeta_{i,i})\kappa_{\tau,\bar{\tau}}^{i,j}\nonumber \\
 & \quad\quad+\sum_{i,j\in\mathcal{M}}\pi_{i}(\ln\zeta_{j,i})\frac{1}{\bar{\tau}}\sum_{l=1}^{\bar{\tau}}p_{i,j}^{(l-1)}\nonumber \\
 & \quad=\sum_{i\in\mathcal{M}}\pi_{i}(\ln\zeta_{i,i})\sum_{j\in\mathcal{M}}\kappa_{\tau,\bar{\tau}}^{i,j}\nonumber \\
 & \quad\quad+\sum_{i,j\in\mathcal{M}}\pi_{i}(\ln\zeta_{j,i})\frac{1}{\bar{\tau}}\sum_{l=1}^{\bar{\tau}}p_{i,j}^{(l-1)}\nonumber \\
 & \quad=\frac{1}{\bar{\tau}}\sum_{l=1}^{\bar{\tau}}\sum_{i,j\in\mathcal{M}}\pi_{i}p_{i,j}^{(l-1)}\ln\zeta_{j,i},\quad\tau\leq\bar{\tau}.\label{eq:zetaineq}
\end{align}
Finally, it follows from (\ref{eq:newcondformondeobservationinstants})
and (\ref{eq:zetaineq}) that 
\begin{align}
 & \sum_{\tau\in\mathbb{N}}\mu_{\tau}\sum_{l=1}^{\tau}\sum_{i,j\in\mathcal{M}}\pi_{i}p_{i,j}^{(l-1)}\ln\zeta_{j,i}\nonumber \\
 & \quad=\sum_{\tau\in\mathbb{N}}\mu_{\tau}\tau\big(\frac{1}{\tau}\sum_{l=1}^{\tau}\sum_{i,j\in\mathcal{M}}\pi_{i}p_{i,j}^{(l-1)}\ln\zeta_{j,i}\big)\nonumber \\
 & \quad\leq\sum_{\tau\in\mathbb{N}}\mu_{\tau}\tau\big(\frac{1}{\bar{\tau}}\sum_{l=1}^{\bar{\tau}}\sum_{i,j\in\mathcal{M}}\pi_{i}p_{i,j}^{(l-1)}\ln\zeta_{j,i}\big).\label{eq:musumineq}
\end{align}
Note that (\ref{eq:newzetacond}) and (\ref{eq:musumineq}) imply
(\ref{eq:condzeta}). Hence, the result follows from Theorem~\ref{maintheorem}.
\end{proof}

Conditions of Theorem~\ref{relaxtautheorem} can be utilized for
assessing stability of a switched stochastic control system, even
if exact knowledge of the distribution $\mu:\mathbb{N}\to[0,1]$ is
not available. Note that the requirement on the knowledge of $\mu:\mathbb{N}\to[0,1]$
is relaxed in Theorem~\ref{relaxtautheorem} by imposing other conditions
on the mode-transition probability matrix $P\in\mathbb{R}^{M\times M}$
and the scalars $\zeta_{i,j}\in(0,\infty)$, $i,j\in\mathcal{M}$. 

\end{rem}

\section{Illustrative Numerical Examples \label{sec:Illustrative-Numerical-Example}}

In this section we provide numerical examples to demonstrate the results
presented in this paper. 

\begin{exmp}Consider the switched stochastic system (\ref{eq:control-system})
with $M=2$ modes described by the subsystems matrices 
\begin{align*}
A_{1}=\left[\begin{array}{cc}
0 & 1\\
1.6 & -0.3
\end{array}\right] & \,,\quad A_{2}=\left[\begin{array}{cc}
0 & 1\\
-0.5 & 1.4
\end{array}\right],
\end{align*}
 $B_{1}=[0,\,1]^{\mathrm{T}}$, and $B_{2}=[0,\,-1]^{\mathrm{T}}$.
The mode signal $\{r(k)\in\mathcal{M}\triangleq\{1,2\}\}_{k\in\mathbb{N}_{0}}$
of the switched system is assumed to be an aperiodic and irreducible
Markov chain characterized by the transition probabilities $p_{1,2}=p_{2,1}=0.3$
and $p_{1,1}=p_{2,2}=0.7$. The invariant distribution for $\{r(k)\in\mathcal{M}\triangleq\{1,2\}\}_{k\in\mathbb{N}_{0}}$
is given by $\pi_{1}=\pi_{2}=0.5$. Moreover, $\mu:\mathbb{N}\to[0,1]$,
according to which the lengths of intervals between consecutive mode
observation instants are distributed, is assumed to be given by $\mu_{\tau}=(1-\theta)^{\tau-1}\theta$,
$\tau\in\mathbb{N}$, with $\theta=0.3$. In this case, at each time
step $k\in\mathbb{N}$, the mode may be observed with probability
$\theta=0.3$ (see Remark~\ref{remarkfortheexample}). 

Note that 
\begin{align}
\tilde{R} & =\left[\begin{array}{cc}
3.0143 & -0.1485\\
-0.1485 & 1.5280
\end{array}\right],
\end{align}
$L_{1}=\left[-3.5326\,\,\,0.9608\right]$, $L_{2}=\left[-3.0029\,\,\,1.8284\right]$,
and the scalars $\zeta_{1,1}=0.7$, $\zeta_{1,2}=1.8$, $\zeta_{2,1}=2$,
and $\zeta_{2,2}=0.8$ satisfy (\ref{eq:condp}) and (\ref{eq:condzeta}).
Now, it follows from Theorem~\ref{maintheorem} that the proposed
control law (\ref{eq:control-law}) with feedback gain matrices 
\begin{align}
K_{1} & =L_{1}\tilde{R}^{-1}=\left[-1.1465\,\,\,0.5174\right],\label{eq:examplek1}\\
K_{2} & =L_{2}\tilde{R}^{-1}=\left[-0.9718\,\,\,1.1021\right],\label{eq:examplek2}
\end{align}
guarantees almost sure asymptotic stability of the closed-loop switched
stochastic system (\ref{eq:control-system}), (\ref{eq:control-law}). 

Sample paths of the state $x(k)$ and the control input $u(k)$ (obtained
with initial conditions $x(0)=\left[1,\,-1\right]^{\mathrm{T}}$ and
$r(0)=1$) are shown in Figures~\ref{Flo:x1} and \ref{Flo:u1}.
Furthermore, Figure~\ref{Flo:rsn1} shows a sample path of the actual
mode signal $r(k)$ and its sampled version $\sigma(k)$. Figures~\ref{Flo:x1}--\ref{Flo:rsn1}
indicate that our proposed control framework guarantees stabilization
even for the case where operation mode of the switched system is observed
only at random time instants. 

\begin{figure}[t]
\begin{center}\includegraphics[width=0.9\columnwidth]{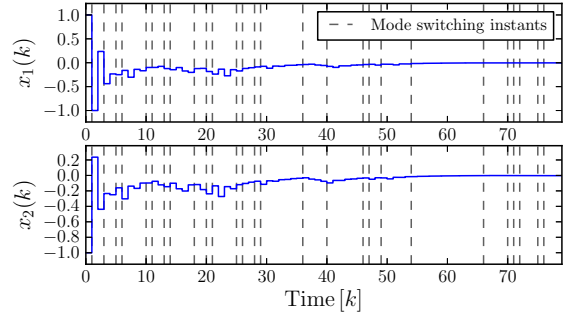}\end{center}
\vskip -5pt\protect\caption{State trajectory versus time}
\label{Flo:x1}
\end{figure}

\begin{figure}[t]
\begin{center}\includegraphics[width=0.9\columnwidth]{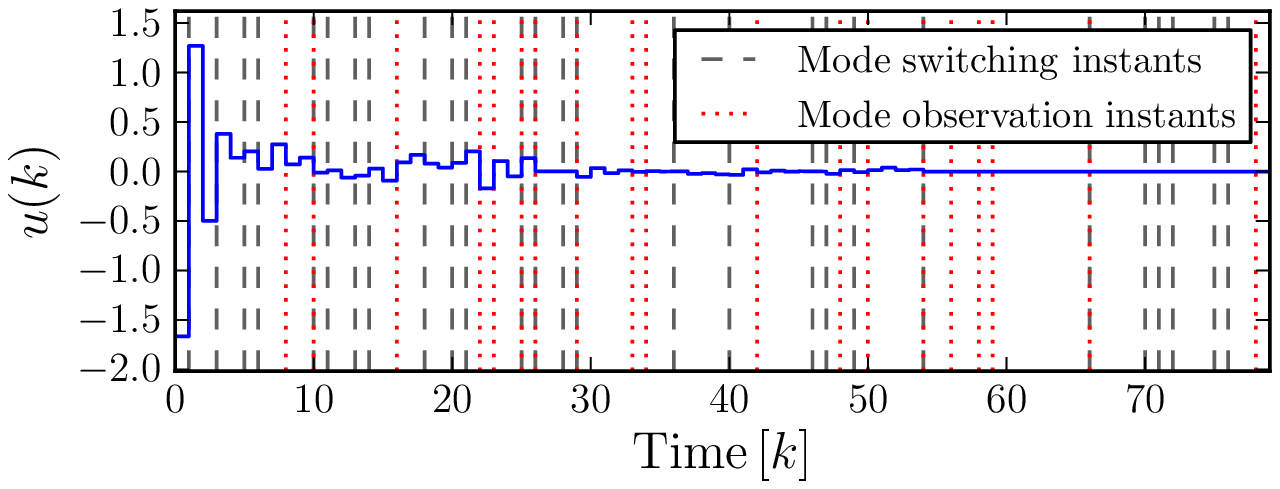}\end{center}
\vskip -5pt\protect\caption{Control input versus time }
\label{Flo:u1}
\end{figure}

\begin{figure}[t]
\begin{center}\includegraphics[width=0.87\columnwidth]{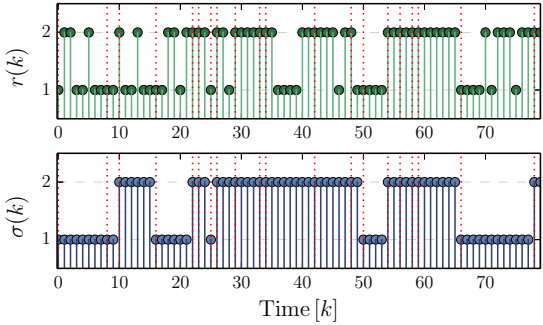}\end{center}
\vskip -5pt\protect\caption{Actual mode $r(k)$ and sampled mode $\sigma(k)$ }
\label{Flo:rsn1}
\end{figure}

The control law (\ref{eq:control-law}) with feedback gain matrices
(\ref{eq:examplek1}) and (\ref{eq:examplek2}) guarantee stabilization
of the closed-loop system with random mode observations characterized
by distribution $\mu_{\tau}=(1-\theta)^{\tau-1}\theta$ with $\theta=0.3$.
Note that for each time step, $\theta$ represents the probability
of mode information being available for control purposes. In order
to investigate conservativeness of our results, we search all values
of parameter $\theta$ for which the control law (\ref{eq:control-law})
with feedback gains (\ref{eq:examplek1}) and (\ref{eq:examplek2})
achieve stabilization. To this end, first, we search values of $\theta$
such that there exist a positive-definite matrix $\tilde{R}$, and
scalars $\zeta_{i,j}$, $i,j\in\mathcal{M}$ that satisfy conditions
(\ref{eq:condp}) and (\ref{eq:condzeta}) of Theorem~\ref{maintheorem}
with $L_{1}=K_{1}\tilde{R}$ and $L_{2}=K_{2}\tilde{R}$, where $K_{1}$
and $K_{2}$ are given by (\ref{eq:examplek1}) and (\ref{eq:examplek2}).
We find that for parameter values $\theta\in[0.2,1]$, conditions
(\ref{eq:condp}) and (\ref{eq:condzeta}) are satisfied. Hence Theorem~\ref{maintheorem}
guarantees stabilization for the case where parameter $\theta$ is
inside the range $[0.2,1]$. On the other hand, through repetitive
numerical simulations we observe that the states of the closed-loop
system converge to the origin in fact for a larger range of parameter
values ($\theta\in[0.12,1]$), which indicate some conservativeness
in the conditions of Theorem~\ref{maintheorem} (see Remark~\ref{conservatism-remark}). 

\end{exmp}

\begin{exmp}Consider the switched stochastic system (\ref{eq:control-system})
with $M=3$ modes described by the subsystems matrices 
\begin{align*}
A_{1}=\left[\begin{array}{cc}
0 & 1\\
1.5 & 0.5
\end{array}\right] & ,\,\,\, A_{2}=\left[\begin{array}{cc}
0 & 1\\
1 & 0.5
\end{array}\right],\,\,\, A_{3}=\left[\begin{array}{cc}
0 & -1\\
1.1 & 1.2
\end{array}\right],
\end{align*}
 $B_{1}=[0,\,1]^{\mathrm{T}}$, $B_{2}=[0,\,0.2]^{\mathrm{T}}$, and
$B_{3}=[0,\,0.7]^{\mathrm{T}}$. The mode signal $\{r(k)\in\mathcal{M}\triangleq\{1,2,3\}\}_{k\in\mathbb{N}_{0}}$
of the switched system is assumed to be an aperiodic and irreducible
Markov chain characterized by the transition matrix $P$ with entries
$p_{i,i}=0.6$, $i\in\mathcal{M}$, and $p_{i,j}=0.2$, $i\neq j$,
$i,j\in\mathcal{M}$. The invariant distribution for $\{r(k)\in\mathcal{M}\triangleq\{1,2,3\}\}_{k\in\mathbb{N}_{0}}$
is given by $\pi_{1}=\pi_{2}=\pi_{3}=\frac{1}{3}$. Furthermore, note
that the transition matrix $P$ possesses positive real eigenvalues
$0.4$ (with algebraic multiplicity $2$) and $1$. The lengths of
intervals between consecutive mode observation instants are assumed
to be uniformly distributed over the set $\{2,3,4,5\}$ (see Remark~\ref{remarkforexample2}).
In other words, the distribution $\mu:\mathbb{N}\to[0,1]$ is assumed
to be given by (\ref{eq:uniformmu}) with $\tau_{\mathrm{L}}=2$ and
$\tau_{\mathrm{H}}=5$. Note that for this example the mode observation
instants $t_{i}$, $i\in\mathbb{N}_{0}$, satisfy (\ref{eq:newcondformondeobservationinstants})
with $\bar{\tau}=5$. 

In this example, we will utilize Theorem~\ref{relaxtautheorem} for
the case where the upper-bounding constant $\bar{\tau}=5$ is known,
but the exact knowledge of the distribution $\mu:\mathbb{N}\to[0,1]$
is not available (see Remark~\ref{relaxtautheorem}). Specifically,
note that 
\begin{align}
\tilde{R} & =\left[\begin{array}{cc}
2.6465 & -0.7851\\
-0.7851 & 1.2568
\end{array}\right],
\end{align}
 $L_{1}=\left[-3.5858\,\,\,\,\,0.1413\right]$, $L_{2}=\left[-4.7066\,\,\,\,-0.3329\right]$,
$L_{3}=\left[-3.2532\,\,\,\,-0.3601\right]$, and the scalars $\zeta_{1,1}=0.6$,
$\zeta_{1,2}=1.7$, $\zeta_{1,3}=1.5$, $\zeta_{2,1}=1.6$, $\zeta_{2,2}=0.7$,
$\zeta_{2,3}=2$, $\zeta_{3,1}=2$, $\zeta_{3,2}=2$, and $\zeta_{3,3}=0.5$
satisfy (\ref{eq:condp}), (\ref{eq:zetarelaxingcondition}), and
(\ref{eq:newzetacond}). Therefore, it follows from Theorem~\ref{relaxtautheorem}
that the proposed control law (\ref{eq:control-law}) with feedback
gain matrices $K_{1}=L_{1}\tilde{R}^{-1}=\left[-1.6222\,\,\,-0.9009\right]$,
$K_{2}=L_{2}\tilde{R}^{-1}=\left[-2.2794\,\,\,-1.6888\right]$, $K_{3}=L_{3}\tilde{R}^{-1}=\left[-1.6132\,\,\,-1.2942\right],$
guarantees almost sure asymptotic stability of the closed-loop system
(\ref{eq:control-system}), (\ref{eq:control-law}). 

Figures~\ref{Flo:x2} and \ref{Flo:u2} respectively show sample
paths of the state $x(k)$ and the control input $u(k)$ obtained
with initial conditions $x(0)=\left[1,\,-1\right]^{\mathrm{T}}$ and
$r(0)=1$. Furthermore, a sample path of the actual mode signal $r(k)$
and its sampled version $\sigma(k)$ are shown in Figure~\ref{Flo:rsn2}.
As it is indicated in Figures~\ref{Flo:x2}--\ref{Flo:rsn2}, the
proposed control framework (\ref{eq:control-law}) achieves asymptotic
stabilization of the zero solution. It is important to note that the
feedback gains $K_{1}$, $K_{2}$, and $K_{3}$ are designed by utilizing
Theorem~\ref{relaxtautheorem} without using information on the distribution
$\mu:\mathbb{N}\to[0,1]$. Note that Theorem~\ref{relaxtautheorem}
requires only the knowledge of an upper-bounding constant $\bar{\tau}\in\mathbb{N}$
for the length of intervals between consecutive mode observation instants,
instead of the exact knowledge of $\mu:\mathbb{N}\to[0,1]$. 

\begin{figure}[t]
\begin{center}\includegraphics[width=0.9\columnwidth]{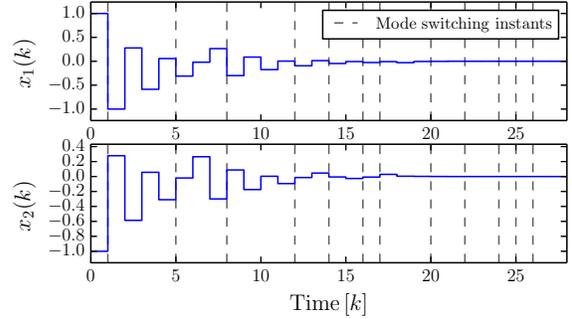}\end{center}
\vskip -7pt\protect\caption{State trajectory versus time}
\label{Flo:x2}
\end{figure}

\begin{figure}[t]
\begin{center}\includegraphics[width=0.9\columnwidth]{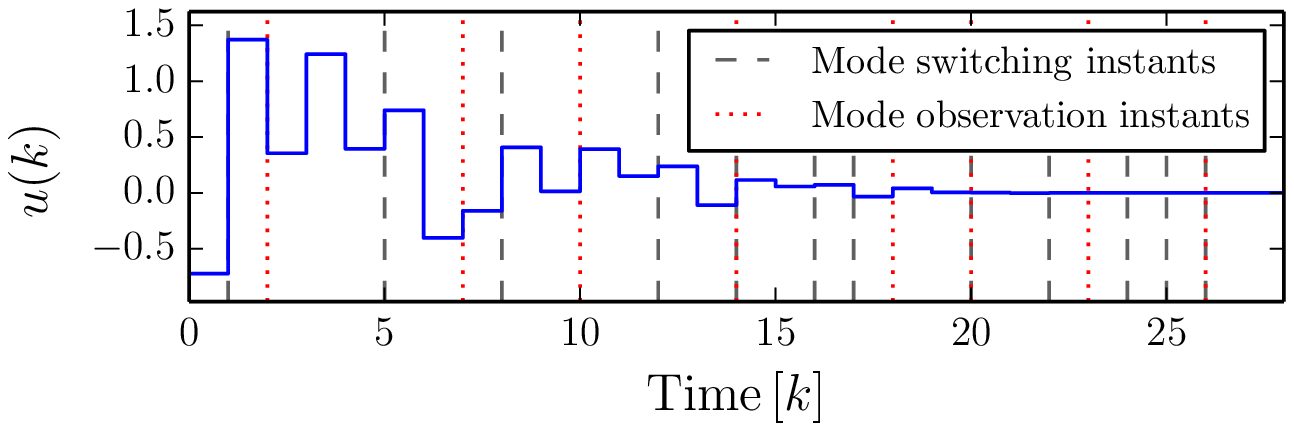}\end{center}
\vskip -7pt\protect\caption{Control input versus time }
\label{Flo:u2}
\end{figure}

\begin{figure}[t]
\begin{center}\includegraphics[width=0.85\columnwidth]{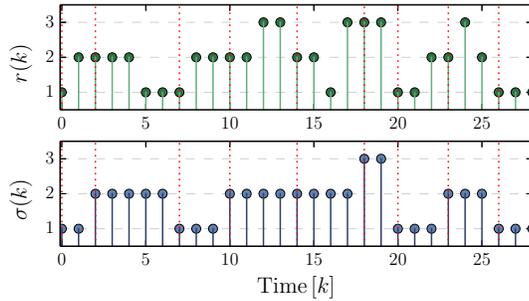}\end{center}
\vskip -7pt\protect\caption{Actual mode $r(k)$ and sampled mode $\sigma(k)$ }
\label{Flo:rsn2}
\end{figure}

\end{exmp}

\section{Conclusion \label{sec:Conclusion}}

We proposed a feedback control framework for stabilization of switched
linear stochastic systems under randomly available mode information.
In this problem setting, information on the active operation mode
of the switched system is assumed to be available for control purposes
only at random time instants. We presented a probabilistic analysis
concerning a sequence-valued stochastic process that captures the
evolution of active operation mode between mode observation instants.
We then used the results of this analysis to obtain sufficient almost
sure asymptotic stability conditions for the zero solution of the
closed-loop system. 

\balance

\bibliographystyle{automatica}
\bibliography{references}

\end{document}